\documentclass[11pt]{article}
\usepackage{fullpage}
\usepackage{float,url}
\usepackage{graphics,amsmath,amssymb,amsfonts}
\usepackage{amsthm}
\usepackage{stackengine}
\usepackage{latexsym}
\usepackage{xcolor}
\usepackage{multirow,multicol}
\usepackage{graphicx}

\newtheorem{theorem}{Theorem}
\newtheorem{corollary}[theorem]{Corollary}
\newtheorem{lemma}[theorem]{Lemma}
\newtheorem{proposition}[theorem]{Proposition}

\newtheorem{remarkthm}[theorem]{Remark}
\newenvironment{remark}{\begin{remarkthm}\normalfont}{\end{remarkthm}}
\newtheorem{problem}{Problem}
\newtheorem{definition}[problem]{Definition}
\newtheorem{openquestionthm}[problem]{Open Question}
\newenvironment{openquestion}{\begin{openquestionthm}\normalfont}{\end{openquestionthm}}

\DeclareMathOperator{\lcm}{lcm}
\def\modd#1 #2{#1\ ({\rm mod}\ #2)}
\usepackage{algorithm,algorithmicx,algpseudocode}
\newcommand{\LineIf}[2]{\State \algorithmicif\ {#1}\ \algorithmicthen\ {#2} \algorithmicend\ \algorithmicif}
\newcommand{\LineIfElse}[3]{\State \algorithmicif\ {#1}\ \algorithmicthen\ {#2} \algorithmicelse\ {#3} \algorithmicend\ \algorithmicif}
\algblockdefx[Choose]{Choose}{EndChoose}[0]{\textbf{choose:}}{\textbf{end choose}}
\algcblockdefx[Choose]{Choose}{OrChoose}{}[0]{\textbf{or:}}{}

\newcommand{\LineChooseThree}[3]{\State \textbf{choose:} {#1}\ \textbf{or:} {#2}\ \textbf{or:} {#3}\ \algorithmicend\ \textbf{choose}}
\renewcommand{\O}{\mathcal{O}}

\title{Existential length universality}

\author{
Pawe{\l} Gawrychowski \\
Institute of Computer Science \\
University of Wroc{\l}aw \\
Wroc{\l}aw, Poland \\
{\tt gawry@cs.uni.wroc.pl} \\
\and
Martin Lange \\
School of Electr. Eng. and Comp. Sc. \\
University of Kassel \\
Kassel, Germany \\
{\tt martin.lange@uni-kassel.de} \\
\and
Narad Rampersad \\
Department of Math/Stats \\
University of Winnipeg \\
515 Portage Ave. \\
Winnipeg, MB R3B 2E9, Canada \\
{\tt narad.rampersad@gmail.com} \\
\and
Jeffrey Shallit \\
School of Computer Science \\
University of Waterloo \\
Waterloo, ON  N2L 3G1, Canada \\
{\tt shallit@cs.uwaterloo.ca} \\
\and 
Marek Szyku{\l}a \\
Institute of Computer Science \\
University of Wroc{\l}aw \\
Wroc{\l}aw, Poland \\
{\tt msz@cs.uni.wroc.pl}}

\begin{document}
\maketitle
\begin{abstract}
We study the following natural variation on the classical universality problem:
given a language $L(M)$ represented by $M$ (e.g., a DFA/RE/NFA/PDA), does there exist an integer $\ell \geq 0$ such that $\Sigma^\ell \subseteq L(M)$?
In the case of an NFA, we show that this problem is NEXPTIME-complete, and the smallest such $\ell$ can be doubly exponential in the number of states.
This particular case was formulated as an open problem in 
2009, and our solution uses a novel and involved construction.
In the case of a PDA, we show that it is recursively unsolvable, while the smallest such $\ell$ is not bounded by any computable function
of the number of states.
In the case of a DFA, we show that the problem is NP-complete,
and $e^{\sqrt{n \log n} (1+o(1))}$ is an asymptotically tight upper bound for the smallest such $\ell$, where $n$ is the number of states.
Finally, we prove that in all these cases, the problem becomes computationally easier when the length $\ell$ is also given in binary in the input: it is polynomially solvable for a DFA, PSPACE-complete for an NFA, and co-NEXPTIME-complete for a PDA.
\end{abstract}
\section{Introduction}

The classical universality problem is the question, for a given language $L$ over an alphabet $\Sigma$, whether $L=\Sigma^*$.
Depending on how $L$ is specified, the complexity of this problem varies.
For example, when $L$ is given as the language accepted by a DFA $M$, the problem is solvable in linear time (reachability of a non-final state) and further is NL-complete \cite{Jones:1975}.
When $L$ is specified by an NFA or a regular expression, it is PSPACE-complete \cite{Aho&Hopcroft&Ullman:1974}. 
When $L$ is specified by a PDA (push-down automaton) or a context-free grammar, the problem is undecidable \cite{Hopcroft&Ullman:1979}.

Studies on universality problems have a long tradition in computer science and still attract much interest.
For instance, the universality has been studied for visibly push-down automata \cite{Alur&Madhusudan:2004}, where the question was shown to be 
decidable in this model (in contrast to undecidability in the ordinary model); timed automata \cite{Bertrand&Bouyer&Brihaye&Stainer:2011}; the language of all prefixes (resp., suffixes, factors, subwords) of the given language \cite{Rampersad&Shallit&Xu:2012}; and recently, for partially (and restricted partially) ordered NFAs \cite{Krotzsch&Masopust&Thomazo:2016}.

In this paper, we study a basic variation of the universality problem, where instead of testing the full language, we ask whether there is a single length that is universal for the language.
We focus on the following two problems:
\begin{problem}[Existential length universality]\label{pbm:elu}
Given a language $L$ represented by a machine $M$ of some type (DFA/RE/NFA/PDA) over an alphabet $\Sigma$ of a fixed size, does there exist an integer $\ell \geq 0$ such that $\Sigma^\ell \subseteq L$?
\end{problem}
\begin{problem}[Specified-length universality]\label{pbm:glu}
Given a language $L$ by a machine $M$ of some type (DFA/RE/NFA/PDA) over an alphabet $\Sigma$ of a fixed size and an integer $\ell$ (given in binary), is $\Sigma^\ell \subseteq L$?
\end{problem}
Furthermore, if such an $\ell$ exists, we are interested in how large the smallest $\ell$ can be.
\begin{definition}
The \emph{minimum universality length} of a language $L$ over an alphabet $\Sigma$ is the smallest integer $\ell \ge 0$ such that $\Sigma^\ell \subseteq L$.
\end{definition}

\subsection{Motivation}

From the mathematical point of view, Problems~\ref{pbm:elu} and~\ref{pbm:glu} are natural variations of universality that surprisingly, to the best of our knowledge, have not been thoroughly investigated in the literature.
Both problems can be seen as an interesting generalization of the famous Chinese remainder theorem to languages, in the sense that given periodicities with multiple periods, we ask where all these periodicities coincide.
Hence, languages stand as succinct representations of integers.
Moreover, both problems are motivated by potential applications in verification listed below.
Finally, the techniques developed to study them are interesting on their own and are likely to find applications elsewhere.
In particular, to solve the case of an NFA, we develop a novel formalism that helps to build NFAs with particular properties.
Indeed, similar constructions to some of the first ingredients in our proof (variables and the Incrementation Gadget), were recently independently discovered to solve the problem of a maximal chain length of the Green relation components of the transformation semigroup of a given DFA \cite{Fleischer&Kufleitner:2017}.

\subsubsection*{Games with imperfect information.}
We consider games with imperfect information on a labeled graph \cite{Doyen&Raskin:2011}, which are used to model, e.g., reactive systems.
In such a game there are two players,
{\it P} modeling the program and {\it E} modeling the environment.
The game starts at the initial vertex.
In one round, {\it P} chooses a label (action) and {\it E} chooses an edge (effect) from the current vertex with this label.
If {\it P} is deterministic and cannot see the choices of {\it E}, then a strategy for {\it P} is just a word over the alphabet of labels.
The game can end under various criteria, and the sequence of the resulting labels determines which player wins.

Under one of the simplest ending criteria, the game lasts for a given number of steps known to {\it P}.
This models the situation when we are interested in the status of the system after some known amount of time.
For example, this occurs if a system must work effectively for a specified duration, but in the end, we must be able to shut it down, which requires that there are no incomplete processes still running inside.
Another example could be a distributed system, where processes have limited possibilities to communicate with the others and are affected by the environment; in general, processes do not know if the system has reached its global goal; hence, for a given strategy, we may need a guarantee of reaching the goal after a known number of rounds, instead of monitoring termination externally.

The main question for such a game is: does {\it P} have a winning strategy?
This is equivalent to asking whether all strategies of {\it P} are losing.
In other words, if $L$ is the language of all losing strategies of {\it P}, then we ask whether all words of the given length are in $L$.
For instance, if the winning criterion for {\it P} is just being in one of the specified vertices, then $L$ is directly defined by the NFA obtained from the graph, where we mark all the non-specified vertices to be final.
We can also ask whether the system is unsafe, i.e., we cannot find a strategy for some number of steps, which corresponds to existential length universality.

One can mention the relationship with the ``firing squad synchronization problem'' \cite{Mazoyer:1987,Moore&Langdon:1968}.
In this classical problem, we are given a cellular automaton of $n$ cells, with one active cell, and the goal is to reach a state in which all cells are simultaneously active. Thinking of an evolving device where the state at time $i$ is represented by the strings of length $i$ in $L$, our problem concerns how many time steps are needed until ``all cells are active'', that is, until all strings of length $i$ are accepted.

\subsubsection*{Formal specifications.}
Our problems are strongly related to a few other questions that can be applied in formal verification, where program correctness is often expressed through inclusion problems \cite{Vardi&Wolper:1994}.
The universality problem is closely related to the inclusion problem: clearly, universality is a special case of inclusion if the underlying language model can express $\Sigma^*$; 
and inclusion can be reduced to universality when this model is closed under unions and includes some (simple) regular languages (possibly folklore, cf.\ \cite{Friedmann&Lange:2012}). 
These reductions carry through to the \emph{given} and the \emph{existential length inclusion} problems.
We can consider the constrained inclusion problems corresponding to the universality problems mentioned above; for instance, \emph{existential length inclusion} asks for two languages $K$ and $L$ whether $K \cap \Sigma^\ell \subseteq L$ for some $\ell$.
Given length inclusion contains the essence of a specialized program verification problem: do all program runs \emph{of a particular duration} adhere to a given specification?
Likewise, existential length universality can be used to check whether there is some number $\ell$ such that running the given program for exactly $\ell$ steps ensures that the specification is met.

Another question of potential interest in program verification is the \emph{bounded-length universality} problem, i.e., whether $\Sigma^{\le \ell} \subseteq L$ for a given $\ell$, resp.\ its inclusion variant.
This could then be used to check whether an implementation meets its specification up to some point, in order to know, for example, whether the safety of a program can be guaranteed for as long as it is terminated externally at some point. 
The complexity of bounded-length universality is the same as that of specified-length universality, which is easy to show by modifying our proofs.

\subsection{Contribution}

In Section~\ref{sec:dfa}, we consider the case where $M$ is a deterministic finite automaton.
Existential length universality is NP-complete, and there exist $n$-state DFAs for which the minimal universality length is of the form $e^{\sqrt{n \log n} (1+o(1))}$, which is the best possible even when the input alphabet is binary.
Specified-length universality is solvable in polynomial time.

In Section~\ref{sec:nfa}, we consider the case where $M$ is a nondeterministic finite automaton.
This particular case was formulated as an open question in May 2009, as mentioned in~\cite{ShallitBC4}, and the solution requires the most involved construction of all problems studied in this paper.
It is easy to show PSPACE-hardness (by modifying the proof from~\cite[Section~10.6]{Aho&Hopcroft&Ullman:1974}) and an NEXPTIME algorithm (by determinization and application of the bounds for DFAs).
We show that, in fact, the problem is NEXPTIME-complete for NFAs.
While to this end, we reduce a standard NEXPTIME-hard problem, designing the reduction requires a non-trivial insight into the structure of the problem and quite a bit of work.
We start with designing a particular intermediate formalism that forms a programming language and makes designing the reduction easier.
Also using the method, we show that the minimum universality length in the case of an NFA can be doubly exponential.
As an auxiliary question in the reduction, we formulate a problem of a specific subset of first-order logic and the Presburger arithmetic.
In the final reduction, we prove that the existential divisibility problem is NEXPTIME-complete.

In Section~\ref{sec:re}, we consider the case when $M$ is a regular expression.
Existential length universality is PSPACE-hard and in NEXPTIME, and there are examples where the minimal universality length is exponential.
The question about the exact complexity class remains open in this case.
Specified-length universality for REs and also for NFAs is PSPACE-complete, which follows from modifying the PSPACE-hardness proof of the usual universality \cite[Section~10.6]{Aho&Hopcroft&Ullman:1974}.
The difficulty in finding a solution is that regular expressions are a more restrictive representation than NFAs, so if the problem is NEXPTIME-complete, the proof requires an even more involed construction than that for NFAs.

Finally, in Section~\ref{sec:pda}, we study the problem where $M$ is a pushdown automaton.
Here, existential length universality is recursively unsolvable, while the minimal universality length grows faster than any computable function, which follows from the undecidability of the universality of a PDA.
On the other hand, specified-length universality is co-NEXPTIME-complete, which we show by another original construction\footnote{We thank an anonymous referee for pointing out that coNEXPTIME-hardness of given-length universality for PDA could also be obtained through a modification of the proof of \cite[Theorem~8.1]{Zetzsche:2016}.}, though less involved than that for NFAs, reducing from an exponential variant of the tiling problem \cite{Boas:1997}.

Our results are summarized in Table~\ref{tab:summary}.

\begin{table}\renewcommand{\arraystretch}{1.2}\small
\caption{Computational complexity of universality problems.}\label{tab:summary}
\newcommand{\rowt}[1]{\multirow{2}{*}{#1}}
\newcommand{\colt}[1]{\multicolumn{2}{c|}{#1}}
\begin{tabular}{|l|c|c|c|c|}\hline
                                &\textbf{DFA}     &\textbf{RE}     &\textbf{NFA}        &\textbf{PDA}         \\\hline
Universality                    & NL-c            & PSPACE-c       & PSPACE-c           & Undecidable         \\\hline
Existential length universality & \rowt{NP-c}     & PSPACE-hard,   & \rowt{NEXPTIME-c}  & \rowt{Undecidable}  \\
(Problem~\ref{pbm:elu})         &                 & in NEXPTIME    &                    &                     \\\hline
Specified-length universality   & \rowt{in PTIME} & \rowt{PSPACE-c}& \rowt{PSPACE-c}    & \rowt{co-NEXPTIME-c}\\
(Problem~\ref{pbm:glu})         &                 &                &                    &                     \\\hline
Minimal universality length     & subexponential  & \emph{open}    & doubly exponential & uncountable         \\\hline
\end{tabular}
\end{table}

While for proving hardness we use larger alphabets than binary, a standard binarization applies to our problems, so all the complexity results remain valid when the input alphabet is binary.

\begin{lemma}\label{lem:alphabet_binary}
Let $\Sigma$ be an alphabet of size $k \ge 2$.
\begin{enumerate}
\item For an $M$ being a DFA/NFA/PDA with $n$ states over $\Sigma$, we can construct in polynomial time a DFA (NFA, PDA, respectively) $M'$ over a binary alphabet with $(2^{\lceil \log_2 k\rceil}-1)n$ states such that its minimal universality length is equal to $\ell\cdot\lceil \log_2 k\rceil$ where $\ell$ is the minimal universality length for $M$, or it does not exist if the minimal universality length does not exist for $M$.
\item For a regular expression $M$ with $n$ input symbols over $\Sigma$, we can construct in polynomial time a regular expression $M'$ over a binary alphabet with at most $(2\lceil\log_2 k\rceil-1)n$ input symbols such that its the minimal universality length is equal to $\ell\cdot\lceil\log_2 k\rceil$ where $\ell$ is the minimal universality length for $M$, or it does not exist if the minimal universality length does not exist for $M$.
\end{enumerate}
\end{lemma}
\begin{proof}
(1) We create $M'$ over the alphabet $\{0,1\}$ by replacing every state with a full binary tree of height $\lceil\log_2 k\rceil-1$ (so with $2^{\lceil\log_2 k\rceil}-1$ states).
The transitions are set in the way that for an $i$'th letter $a_i \in \Sigma$, the binary representation of $i$ acts on the roots of the states in $M'$ as $a_i$ in $M$ on the corresponding states. The remaining binary representations just duplicate the action of any other one.
The final states of $M'$ are the roots of the trees corresponding to the final states of $M$.

(2) Similarly, a regular expression can be converted to one regular expression $M'$ over $\{0,1\}$ by replacing each $i$'th letter $a_i \in \Sigma$ with either its binary representation or the sum (union) of two binary representations, which all are of length $\lceil \log_2 k\rceil$. Since the number of representations $2^{\lceil\log_2 k\rceil}$ is smaller than $2k$, in this way all of them can be assigned to some letter.

It is easy to observe that for both constructions $M'$, all words of length $\ell$ are accepted by $M$ if and only if all words of length $\ell\cdot\lceil \log_2 k\rceil$ are accepted by $M'$.
\end{proof}

\section{The case where $M$ is a DFA}\label{sec:dfa}

\begin{theorem}\label{thm:DFA_given_length}
Specified-length universality (Problem~\ref{pbm:glu}) for DFAs is solvable in polynomial time (in the size of $M$ and $\log \ell$).
\end{theorem}
\begin{proof}
Start with a DFA $M = (Q, \Sigma, \delta, q_0, F)$ that we assume to be complete (that is, $\delta(q,a)$ is defined for all $q \in Q$ and $a \in \Sigma$).

From $M$ create a unary NFA $M'$ that is defined by taking every transition of $M$ labeled with an input letter and replacing that letter with the single letter $a$.
Now it is easy to see that $M$ accepts all strings of length $\ell$ if and only if every path of length $\ell$ in $M'$, starting with its initial state, ends in a final state.

Now create a Boolean matrix $B$ with the property that there is a $1$ in row $i$ and column $j$ if and only if $M'$ has a transition from $q_i$ to $q_j$, and $0$ otherwise.
It is easy to see that $M'$ has a path of length $\ell$ from state $q_i$ to state $q_j$ if and only if $B^\ell$ has a $1$ in row $i$ and column $j$, where by $B^\ell$ we mean the Boolean power of the matrix $B$.

So $M$ accepts all strings of length $\ell$ if and only if every $1$ in row $0$ (corresponding to $q_0$) of $B^\ell$ occurs only in columns corresponding to the final states of $M'$.
Hence, to verify that $M$ is length universal for some $\ell$, we simply guess $\ell$, compute $B$, raise $B$ to the $\ell$'th Boolean power using the
usual ``binary method'' or ``doubling up'' trick, and check the positions of the $1$'s in row $0$.
This can be done in polynomial time provided $\ell$ is exponentially bounded in magnitude.
\end{proof}

\begin{theorem}\label{thm:DFA_inNP}
Existential length universality (Problem~\ref{pbm:elu}) for DFAs is in NP.
\end{theorem}
\begin{proof}
Since for a given $\ell$ the problem can be polynomially solved, it is enough to guess $\ell$ and check.
To see that $\ell$ is exponentially bounded, we can argue that $\ell \leq 2^{|Q|^2}$.
This follows trivially because our matrix from the proof of Theorem~\ref{thm:DFA_given_length} is of dimension $|Q| \times |Q|$; after we see all $2^{|Q|^2}$ different powers, we have
seen all we can see, and the powers must cycle after that.
\end{proof}

Actually, we can do even better than the $2^{|Q|^2}$ bound in the proof.
It is an old result of Rosenblatt~\cite{Rosenblatt:1957} that powers of a $t \times t$ Boolean matrix are ultimately periodic with preperiod of size $\O(t^2)$ and period of size at most $e^{\sqrt{t \log t}(1+o(1))}$.
Thus we have

\begin{theorem}\label{thm:DFA_length_upper_bound}
Let $M$ be a DFA with $n$ states.
If there exists a minimal universality length $\ell$ for $M$, then $\ell \leq e^{\sqrt{n \log n}(1+o(1))}$.
\end{theorem}

The upper bound in the previous result is tight, even for a binary alphabet.

\begin{theorem}\label{thm:DFA_length_lower_bound}
For each sufficiently large $n$, there exists a binary DFA $M$ with $n$ states for which the minimal universality length $\ell$ is $\geq e^{\sqrt{n \log n}(1+o(1))}$.
\end{theorem}
\begin{proof}
First we construct our DFA with $t$ input symbols:
There is a non-final initial state with transitions out on symbols $a_1,a_2,\ldots,a_t$ to cycles of size $p(1),p(2),\ldots,p(t)$, respectively, where $p(i)$ is the $i$'th prime number.
The transitions on each cycle are on all symbols of the alphabet.
Inside each cycle, all states are not final, except the state immediately before the state with an incoming transition from the initial state, which is final.
This DFA accepts the language 
$$\bigcup_{1 \leq i \leq t} a_i (\Sigma^{p(i)})^* \Sigma^{p(i)-1}.$$
For each length $\ell' < p(1) p(2) \cdots p(t) - 1$, there exists a prime number $p(i)$, $1 \leq i \leq t$, such that $\ell' \not\equiv \modd{p(i)-1} {p(i)}$,
so no string in $a_i \Sigma^{\ell'}$ is accepted.
However, for $\ell = p(1) p(2) \cdots p(t)$, all strings of length $\ell$ are accepted.

Now we convert the DFA to a binary DFA over $\{0,1\}$ by replacing the initial state with a full binary tree of height $h=\lceil\log_2 t \rceil-1$, so that the binary representation of $i$ maps the initial state (root) to the state from the cycle of length $p(i)$ (cf.\ Lemma~\ref{lem:alphabet_binary}).
This DFA has $p(1) + p(2) + \cdots + p(t) + 2^{h+1}-1$ states, and the least $\ell$ for which all strings of length $\ell$ is accepted is $p(1) p(2) \cdots p(t) + h-1$. 

From the prime number theorem we know that
$$p(1) + p(2) + \cdots + p(t) + \O(t) \sim \frac{1}{2} t^2 \log t,$$
(see, for example, \cite[p.~29]{Bach&Shallit:1996}) and
$$p(1) p(2) \cdots p(t) + \O(t) \sim e^{t \log t (1 + o(1))}$$
(see, for example, \cite[Theorem~4.4]{Apostol:1976}).
\end{proof}

\begin{theorem}\label{thm:DFA_NP-hard}
Existential length universality (Problem~\ref{pbm:elu}) is NP-hard for DFAs.
\end{theorem}
\begin{proof}
We reduce from 3-SAT.
Given a formula $\varphi$ in the 3-CNF form we create a DFA $M = M_\varphi$ having the property that there exists an integer $\ell\ge 0$ such that $M$ accepts all strings of length $\ell$ if and only if $\varphi$ has a satisfying assignment.

To do so, we use the ideas from the usual Chinese-remainder-theorem-based proof of the coNP-hardness of deciding if $L(M) = \Sigma^*$ for a unary NFA $M$.
\cite{Stockmeyer&Meyer:1973}.

Suppose there are $t$ variables in $\varphi$, say $v_1,v_2,\ldots,v_t$.
Then satisfying assignments are coded by integers which are either congruent to $0$ or $1$ modulo the first $t$ primes.
Let $p(i)$ be the $i$'th prime number.
If the integer is $1 \bmod p(i)$, then it corresponds to an assignment where $v_i$ is set to $1$;
if the integer is $0 \bmod p(i)$, then it corresponds to an assignment where $v_i$ is set to $0$.
Given a clause $C$ consisting of $3$ literals (variables $v_i, v_j, v_k$ or their negations), all integers corresponding to a satisfying assignment of this particular clause fall into a number of residue classes modulo $p(i) p(j) p(k)$.

We now construct a DFA with an alphabet $\Sigma$ consisting of the integers $1,2,\ldots,s$, where $s$ is the number of clauses.
The non-final initial state has a transition on each letter of $\Sigma$ to a cycle corresponding to the appropriate clause.
Inside each cycle, the transitions go to the next state on all letters of $\Sigma$.
Each cycle is of size $p(i) p(j) p(k)$, where $v_i, v_j, v_k$ are the variables appearing in the corresponding clause.
The final states in each cycle are the integers, modulo $p(i) p(j) p(k)$, that correspond to assignments satisfying that clause.

We claim this DFA accepts all strings of length $\ell$ if and only if $\ell$ corresponds to a satisfying assignment for $\varphi$.
To see this, note that if the DFA accepts all strings of length $\ell$ for some $\ell$, then the path inside each cycle must terminate at a final state, which corresponds to a satisfying assignment for each clause.
On the other hand, if $\ell$ corresponds to a satisfying assignment, then every string is accepted because it satisfies each clause, and hence corresponds to a path beginning with any clause number and entering the appropriate cycle.

The transformation uses polynomial time because the $i$'th prime number is bounded in magnitude by $\O(i \log i)$ (e.g.,~\cite{Rosser&Schoenfeld:1962}), and each cycle is therefore of size at most $\O((t \log t)^3)$, so the total number of states is $\O(s (t \log t)^3)$.
\end{proof}

\section{The case where $M$ is an NFA}\label{sec:nfa}

The classical universality problem for regular expressions and so for NFAs is known to be PSPACE-complete \cite[Section~10.6]{Aho&Hopcroft&Ullman:1974}.
Also, if the NFA does not accept $\Sigma^*$, then the length of the shortest non-accepted words is at most exponential.
Specified-length universality for NFAs is also PSPACE-complete (Theorem~\ref{thm:re_given_length}).

However, we show that existential length universality is harder: it is NEXPTIME-complete, and there are examples where 
the minimal universality length is approximately doubly exponential in the number of states of the NFA.

We begin with upper bounds, which follow from the results for DFAs.

\begin{proposition}\label{pro:nfa_upper_bound}
Let $M$ be an NFA with $n$ states.
If there exists an $\ell$ such that $M$ accepts all strings of length $\ell$,
then the smallest such $\ell$ is $\leq e^{2^{n/2} \sqrt{n \log {2}}(1+o(1))}$.
\end{proposition}
\begin{proof}
By determinizing $M$ to a DFA with at most $2^n$ states and applying Theorem~\ref{thm:DFA_length_upper_bound} we get
$$\ell \leq e^{\sqrt{2^n \log {2^n}}(1+o(1))} = e^{2^{n/2} \sqrt{n \log {2}}(1+o(1))}.$$
\end{proof}

\begin{proposition}
Existential length universality (Problem~\ref{pbm:elu}) for NFAs is in NEXPTIME.
\end{proposition}
\begin{proof}
We determinize $M$ to a DFA with at most $2^n$ states.
Then by Theorem~\ref{thm:DFA_inNP}, existential length universality is solvable in nondeterministic exponential time.
\end{proof}

The difficult part is to show that existential length universality for NFAs is NEXPTIME-hard.
Note that the usual method which is applied to show PSPACE-hardness of the classical universality problem does not seem enough
in this case and, after a suitable modification, results only in a proof of PSPACE-hardness.
The reason for this difficulty is that, to show NEXPTIME-hardness, we need to be able to construct NFAs whose minimum universality length is larger than exponential, which is a non-trivial task in itself.
The NFA constructed by the reduction must have a polynomial size, whereas to solve an NEXPTIME-complete problem we may need exponential memory.
If the length of a word is superexponential, then some subsets of the states of the NFA must be repeated multiple times.

To overcome this major technical hurdle we construct an NFA with minimum universality length being roughly doubly exponential by using
an indirect approach.
We design the automaton such that there are many exponentially long cycles on subsets of the states and 
it accepts all words only for the lengths that are solutions given by the Chinese Remainder Theorem for these cycle lengths.
By exhibiting a family of NFAs with large minimal universality lengths we show that our construction is essentially tight.
The techniques are rather involved, and hence we first develop an intermediate formalism that will be used for both tasks.
Having developed our formalism, we are able to solve the first task.
To solve the second task, we proceed in two steps.
First, we reduce an auxiliary logic decision problem concerning the divisibility of integers to existential length universality through our formalism.
Second, we reduce a canonical NEXPTIME-complete problem to this divisibility problem.

\subsection{A programming language}\label{subsec:programming_in_NFAs}

We define a simple programming language that will be used to construct NFAs with particular properties in a convenient way.
Our language is nondeterministic, i.e., the programs can admit many possible computations of the same length. 
In contrast to the usual programming languages, we are interested only in this set of admitted computations by the program.

A program will be translated in polynomial time directly to an NFA, or, more precisely, to an extended structure called a \emph{gadget}, which is defined below.
A computation of the program will correspond to a word for the constructed NFA.
If the computation is not admitted, the word will be always accepted.
Otherwise, usually, the word will not be accepted, with some exceptions when we additionally make some states final in the NFA.

\subsubsection{Gadget definition}

Let $m \ge 1$ be a fixed integer.
A \emph{variable} $V$ is a set of states $\{v_1,\ldots,v_m,\bar{v}_1,\ldots,\bar{v}_m\}$.
These states are called \emph{variable states}, and $m$ is the \emph{width} of the variable.
Besides variable states, in our NFAs there will be also \emph{control flow states}, and the unique special final state $q_{\rm acc}$, which will be fixed by all transitions.

A \emph{gadget} $G$ is a 7-tuple $(P^G,\mathcal{V}^G,\Sigma^G,\delta^G,s^G,t^G,F^G)$.
When specifying the elements of such a tuple, we usually omit the superscript if it is clear from the context.
$P$ is a set of control flow states, $\mathcal{V}$ is a set of (disjoint) variables on which the gadget \emph{operates}, $s,t \in P$ are distinguished \emph{start} and \emph{target} control flow states, respectively, and $F \subseteq P$ is a set of final states.
The set of states of $G$ is
$Q = \{q_{\rm acc}\} \cup P \cup \bigcup_{V \in \mathcal{V}} V$.
Then $\delta\colon Q \times \Sigma \to 2^Q$ is the transition function, which is extended to a function $2^Q \times \Sigma^* \to 2^Q$ as usual.
We always have $\delta(q_{\rm acc},a) = \{q_{\rm acc}\}$ for every $a \in \Sigma$.

The NFA of $G$ is $(Q,\Sigma,\delta,s,F \cup \{q_{\rm acc}\})$.
A \emph{configuration} is a subset $C \subseteq Q$.
Given a configuration $C$, we say a state is \emph{active} if it belongs to $C$.
We say a configuration $C$ is \emph{proper} if it does not contain $q_{\rm acc}$.
Given a proper configuration $C$ and a word $w$, we say that $w$ is a \emph{proper computation from $C$} if the obtained configuration after reading $w$ is also proper, i.e., $\delta(C,w)$ is proper.
Therefore, from a non-proper configuration we cannot obtain a proper one after reading any word since $q_{\rm acc}$ is always fixed, and so every non-proper computation from $\{s\}$ is an accepted word by the NFA.

We say that a variable $V$ is \emph{valid} in a configuration $C \subseteq Q$ if for all $1 \le i \le m$, $v_i \in C$ if and only if $\bar{v}_i \notin C$.
In other words, the states $\bar{v}_1,\ldots,\bar{v}_m$ are complementary to the states $v_1,\ldots,v_m$.
A valid variable stores an integer from $\{0,\ldots,2^m-1\}$ encoded in binary; the states $v_1$ and $v_m$ represent the least and the most significant bit, respectively.
Formally, if $V$ is valid in a configuration $C$, then its \emph{value} $V(C)$ is defined as
$$V(C) = \sum_{\substack{1 \le i \le m\\v_i \in V \cap C}} \ 2^{i-1}.$$

We say that a configuration $C$ is \emph{initial} for a gadget $G$ if it is proper, contains the start state $s$ but no other control flow states, and the gadget's variables are valid in $C$ (if not otherwise stated, which is the case for some gadgets).
A \emph{final} configuration is a proper configuration that contains the target state $t$ and no other control flow states.
A \emph{complete computation} is a proper computation from an initial configuration to a final configuration.
Every gadget will possess some properties about its variables and the length of complete computations according to its semantics.
These properties are of the form that, depending on an initial configuration $C$, there exists or not a complete computation of some length from $C$ to a final configuration $C'$, where $C'$ also satisfies some properties.
Usually, proper computations from an initial configuration will have bounded length (but not always, as we will also create cycles).
Also usually, proper configurations will have exactly one active control flow state (with the exception of the Parallel Gadget, introduced later).
If a variable is not required to be valid in $C$, then these properties will not depend on its active states in $C$.

We start from defining \emph{basic} gadgets, which are elementary building blocks, and then we will define \emph{compound} gadgets, which are defined using the other gadgets inside.

\subsubsection{Basic gadgets}
\noindent\textbf{$\bullet$ Selection Gadget.}\\
This gadget is denoted by \Call{Select}{$V$}, where $V$ is a variable.
It allows a nondeterministic selection of an arbitrary value for $V$.
An initial configuration for this gadget does not require that $V$ is valid.
For every integer $c \in \{0,\ldots,2^m-1\}$ and for every initial configuration $C$, there exists a complete computation from $C$ to a final configuration $C'$ such that $V(C')=c$.

\begin{figure}[htb]\centering
\includegraphics[scale=1]{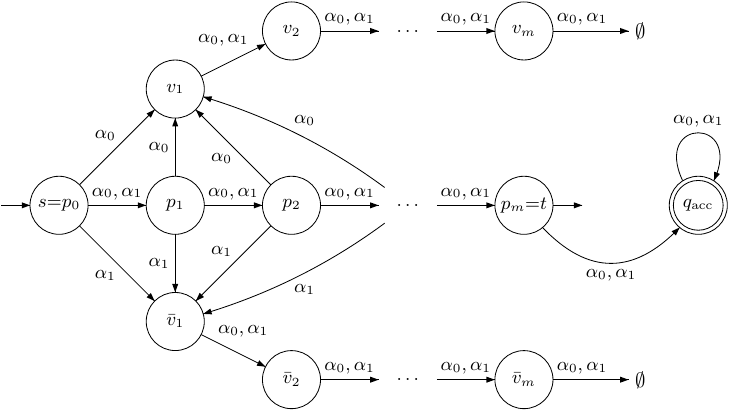}
\caption{Selection Gadget.}\label{fig:selection_gadget}
\end{figure}

The gadget is illustrated in Fig.~\ref{fig:selection_gadget}.
It consists of control flow states $P = \{s=p_0,p_1,\ldots,p_{m-1},p_m=t\}$, one variable $V$, and letters $\Sigma = \{\alpha_0,\alpha_1\}$.
The letters $\alpha_0$ and $\alpha_1$ allow moving the active control flow state over the states $p_0,p_1,\ldots,p_m$ and, at each transition, choosing either $v_1$ or $\bar{v}_1$ to be active.
Also, each $v_i$ and $\bar{v}_i$ are shifted to $v_{i+1}$ and $\bar{v}_{i+1}$, respectively, and both $v_m$ and $\bar{v}_m$ are mapped to no state ($\emptyset$), which ensures that the initial content of $V$ is neglected.
The transitions are defined as follows:
\begin{eqnarray*}
\delta(p_i,\alpha_0) & = & \{p_{i+1},v_1\} \text{ for $i=0,\ldots,m-1$},\\
\delta(p_i,\alpha_1) & = & \{p_{i+1},v'_1\} \text{ for $i=0,\ldots,m-1$},\\
\delta(p_m,\alpha_0) = \delta(p_m,\alpha_1) & = & \{q_{\rm acc}\},\\
\delta(v_i,\alpha_0) = \delta(v_i = \alpha_1) & = & \{v_{i+1}\} \text{ for $i=1,\ldots,m-1$},\\
\delta(\bar{v}_i,\alpha_0) = \delta(\bar{v}_i = \alpha_1) & = & \{\bar{v}_{i+1}\} \text{ for $i=1,\ldots,m-1$},\\
\delta(v_m,\alpha_0) = \delta(\bar{v}_m = \alpha_1) & = & \emptyset.
\end{eqnarray*}
Note that a word $w = \alpha_{b_1} \ldots \alpha_{b_m}$, for $b_i \in \{0,1\}$, sets the value of the variable to $\sum_{1 \le i \le m} 2^{i-1} b_i$.

The semantic properties are summarized in the following
\begin{lemma}\label{lem:selection_gadget}
Let $C$ be an initial configuration for the Selection Gadget \Call{Select}{$V$}.
For every value $c \in \{0,\ldots,2^m-1\}$, there exists a complete computation in $\Sigma^m$ from $C$ to a proper configuration $C'$ such that $V(C')=c$.
Every complete computation has length $m$, and every longer computation is not proper.
\end{lemma}
\begin{proof}
This follows from the construction in a straightforward way.
After reading a word $w \in \{\alpha_0,\alpha_1\}^m$ we get that $t$ is active, and for every $i=1,\ldots,m$ either $v_i$ or $\bar{v}_i$ is active, depending on the $i$'th letter.
Thus, for every value of $V$ there is a unique word $w \in \{\alpha_0,\alpha_1\}^m$ resulting in setting this value.
If $w$ is shorter than $m$, then $t$ cannot become active, since the shortest path from $s$ to $t$ has length $m$.
Longer computations are not proper since both letters map $t$ to $q_{\rm acc}$.
\end{proof}

\noindent\textbf{$\bullet$ Equality Gadget.}\\
This gadget is denoted by $U=V$, where $U$ and $V$ are two distinct variables.
It checks if the values of the valid variables $U$ and $V$ are equal in the initial configuration.
If so, the gadget admits a complete computation, which is of length $m$; otherwise, every word of length at least $m$ is a non-proper computation.

\begin{figure}[htb]\centering
\includegraphics[scale=1]{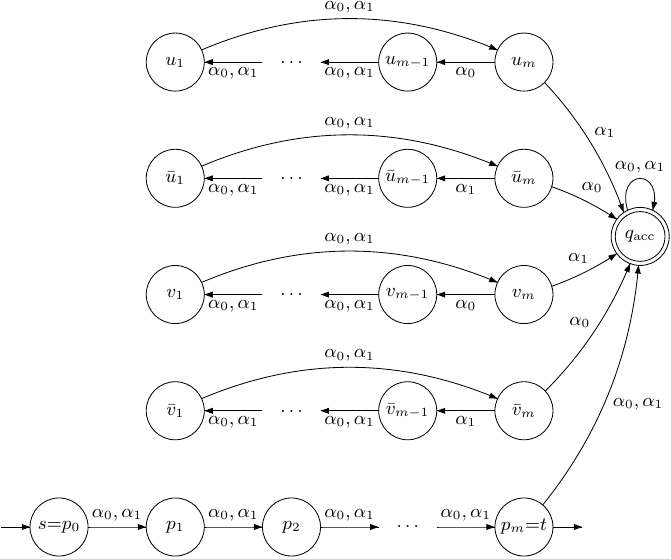}
\caption{Equality Gadget.}\label{fig:equality_gadget}
\end{figure}

The gadget is illustrated in Fig.~\ref{fig:equality_gadget}.
It consists of control flow states $P = \{s=p_0,p_1,\ldots,p_m=t\}$, two variables $U$ and $V$, and letters $\Sigma = \{\alpha_0,\alpha_1\}$.
The letters $\alpha_0$ and $\alpha_1$ allow moving the active control flow state over the states $s=p_0,p_1,\ldots,p_m=t$ and, at each transition, checking if the corresponding positions of $U$ and $V$ agree.
The transitions are defined in the same way for both variables, and they cyclically shift their states.
\begin{eqnarray*}
\delta(v_1,\alpha_0) = \delta(v_1 = \alpha_1) & = & \{v_m\},\\
\delta(v_i,\alpha_0) = \delta(v_i = \alpha_1) & = & \{v_{i-1}\} \text{ for $i=1,\ldots,m-1$},\\
\delta(v_m,\alpha_0) & = & \{v_{m-1}\},\\
\delta(v_m,\alpha_1) & = & \{q_{\rm acc} \},\\
\delta(\bar{v}_1,\alpha_0) = \delta(\bar{v}_1 = \alpha_1) & = & \{\bar{v}_m\},\\
\delta(\bar{v}_i,\alpha_0) = \delta(\bar{v}_i = \alpha_1) & = & \{\bar{v}_{i-1}\} \text{ for $i=1,\ldots,m-1$},\\
\delta(\bar{v}_m,\alpha_0) & = & \{q_{\rm acc}\},\\
\delta(\bar{v}_m,\alpha_1) & = & \{\bar{v}_{m-1}\},\\
\delta(u_1,\alpha_0) = \delta(u_1 = \alpha_1) & = & \{u_m\},\\
\delta(u_i,\alpha_0) = \delta(u_i = \alpha_1) & = & \{u_{i-1}\} \text{ for $i=1,\ldots,m-1$},\\
\delta(u_m,\alpha_0) & = & \{u_{m-1}\},\\
\delta(u_m,\alpha_1) & = & \{q_{\rm acc}\},\\
\delta(\bar{u}_1,\alpha_0) = \delta(\bar{u}_1 = \alpha_1) & = & \{\bar{u}_m\},\\
\delta(\bar{u}_i,\alpha_0) = \delta(\bar{u}_i = \alpha_1) & = & \{\bar{u}_{i-1}\} \text{ for $i=1,\ldots,m-1$},\\
\delta(\bar{u}_m,\alpha_0) & = & \{q_{\rm acc}\},\\
\delta(\bar{u}_m,\alpha_1) & = & \{\bar{u}_{m-1}\},\\
\delta(p_i,\alpha_0) = \delta(p_i,\alpha_1) & = & \{p_{i+1}\} \text{ for $i=0,\ldots,m-1$},\\
\delta(p_m,\alpha_0) = \delta(p_m,\alpha_1) & = & \{q_{\rm acc}\}.
\end{eqnarray*}

\begin{lemma}\label{lem:equality_gadget}
Let $C$ be an initial configuration with valid variables $U$ and $V$ for the Equality Gadget $U=V$.
When $U(C)=V(C)$, there exists a complete computation in $\Sigma^m$ from $C$ to a proper configuration $C'$.
Moreover, every complete computation has length $m$ and is such that $U(C')=U(C)$ and $V(C')=V(C)$.
Longer computations are not proper, and when $U(C) \neq V(C)$, every computation of length at least $m$ is not proper.
\end{lemma}
\begin{proof}
After reading a word $w \in \{\alpha_0,\alpha_1\}^m$ we get that $t$ is active and no shorter word has this property.
If $u_i \in C$ then the $(i+1)$'st letter of $w$ (or first letter if $i=m$) must be $\alpha_0$.
Similarly, if $\bar{u}_i \in C$ then the $(i+1)$'st letter of $w$ (or first if $i=m$) must be $\alpha_1$.
The same holds for the states of $V$.
Thus, if $u_i \in C$ and $\bar{v}_i \in C$, or if $\bar{u}_i \in C$ and $v_i \in C$, then $q_{\rm acc}$ cannot be avoided by any such a word $w$.
Otherwise, there exists a unique word $w$ of length $m$ such that $t$ becomes active and $q_{\rm acc}$ does not.
\end{proof}

\noindent\textbf{$\bullet$ Inequality Gadget.}\\
This gadget is denoted by $U \neq V$, where $U$ and $V$ are two distinct variables.
It is similar to the Equality Gadget and checks if the values of valid variables $U$ and $V$ are different.

\begin{figure}[htb]\centering
\includegraphics[scale=1]{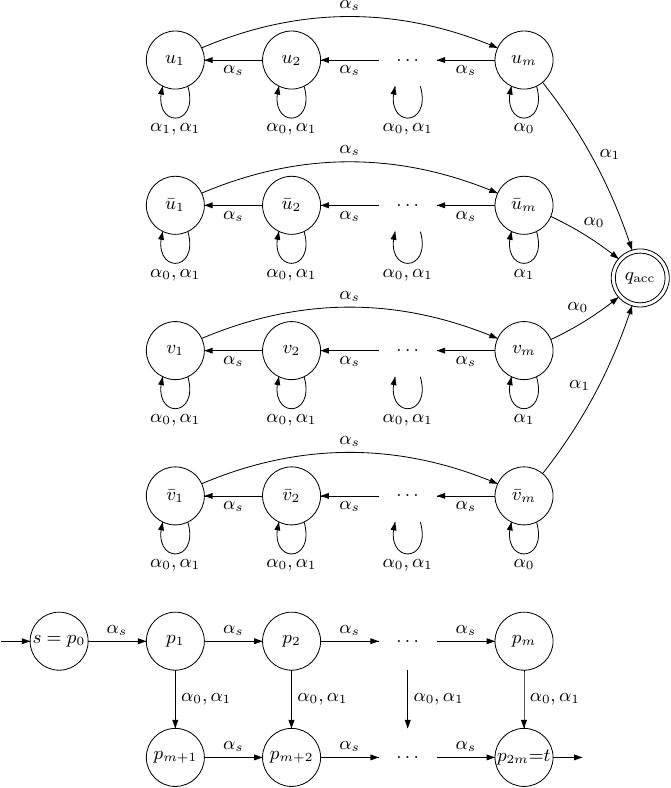}
\caption{Inequality Gadget. All omitted transitions go to $q_{\rm acc}$.}\label{fig:inequality_gadget}
\end{figure}

The gadget is illustrated in Fig.~\ref{fig:inequality_gadget}.
It consists of control flow states $P = \{s=p_0,p_1,\ldots,p_{2m}=t\}$, two variables $U$ and $V$, and letters $\Sigma = \{\alpha_0,\alpha_1,\alpha_s\}$.
The letter $\alpha_s$ allows moving the active control flow state, and it cyclically shifts the states of both variables.
Its transition function is defined as follows:
\begin{eqnarray*}
\delta(u_1,\alpha_s) & = & \{u_m\},\\
\delta(u_i,\alpha_s) & = & \{u_{i-1}\} \text{ for $i=1,\ldots,m$},\\
\delta(\bar{u}_1,\alpha_s) & = & \{\bar{u}_m\},\\
\delta(\bar{u}_i,\alpha_s) & = & \{\bar{u}_{i-1}\} \text{ for $i=1,\ldots,m$},\\
\delta(v_1,\alpha_s) & = & \{v_m\},\\
\delta(v_i,\alpha_s) & = & \{v_{i-1}\} \text{ for $i=1,\ldots,m$},\\
\delta(\bar{v}_1,\alpha_s) & = & \{\bar{v}_m\},\\
\delta(\bar{v}_i,\alpha_s) & = & \{\bar{v}_{i-1}\} \text{ for $i=1,\ldots,m$},\\
\delta(p_i,\alpha_s) & = & \{p_{i+1}\} \text{ for $i=0,\ldots,m-1$},\\
\delta(p_{m+i},\alpha_s) & = & \{p_{m+i+1}\} \text{ for $i=1,\ldots,m-1$},\\
\delta(p_m,\alpha_s) = \delta(p_{2m},\alpha_s) & = & \{q_{\rm acc}\}.
\end{eqnarray*}
At some point, at the position where the variables differ, either $\alpha_0$ or $\alpha_1$ can be applied, which also switches the active control flow state from $p_i$ to the corresponding $p_{m+i}$. Their transitions are defined as follows:
\begin{eqnarray*}
\delta(u_i,\alpha_0) = \delta(u_i,\alpha_1) & = & \{u_i\} \text{ for $i=1,\ldots,m-1$},\\
\delta(\bar{u}_i,\alpha_0) = \delta(\bar{u}_i,\alpha_1) & = & \{\bar{u}_i\} \text{ for $i=1,\ldots,m-1$},\\
\delta(v_i,\alpha_0) = \delta(v_i,\alpha_1) & = & \{v_i\} \text{ for $i=1,\ldots,m-1$},\\
\delta(\bar{v}_i,\alpha_0) = \delta(\bar{v}_i,\alpha_1) & = & \{\bar{v}_i\} \text{ for $i=1,\ldots,m-1$},\\
\delta(u_m,\alpha_0) & = & \{u_m\},\\
\delta(\bar{v}_m,\alpha_0) & = & \{\bar{v}_m\},\\
\delta(\bar{u}_m,\alpha_1) & = & \{\bar{u}_m\},\\
\delta(v_m,\alpha_1) & = & \{v_m\},\\
\delta(\bar{u}_m,\alpha_0) = \delta(v_m,\alpha_0) = \delta(u_m,\alpha_1) = \delta(\bar{v}_m,\alpha_1) & = & \{q_{\rm acc}\},\\
\delta(p_i,\alpha_0) = \delta(p_i,\alpha_1) & = & \{p_{m+i}\} \text{ for $i=1,\ldots,m$},\\
\delta(p_{m+i},\alpha_0) = \delta(p_{m+i},\alpha_1) & = & \{q_{\rm acc}\} \text{ for $i=1,\ldots,m$}.
\end{eqnarray*}

\begin{lemma}\label{lem:inequality_gadget}
Let $C$ be an initial configuration with valid variables $U$ and $V$ for the Inequality Gadget $U \neq V$.
When $U(C) \neq V(C)$, there exists a complete computation in $\Sigma^{m+1}$ from $C$ to a proper configuration $C'$.
Moreover, every complete computation has length $m+1$ and is such that $U(C')=U(C)$ and $V(C')=V(C)$.
Longer computations are not proper, and when $U(C) = V(C)$, every computation of length at least $m+1$ is not proper.
\end{lemma}
\begin{proof}
Consider a word $w$ of length $m+1$.
If $q_{\rm acc}$ does not become active, then by the structure of control flow states, $t$ must become active and $w$ contains exactly $m$ occurrences of $\alpha_s$ and one occurrence of either $\alpha_0$ or $\alpha_1$.
If $\alpha_0$ is the $i$'th letter of $w$ ($2 \le i \le m+1$), then it must be that $u_{i-1} \in C$ and $v_{i-1} \notin C$.
This is dual for $\alpha_1$.
Therefore, there must exist a position at which the variables differ.

Conversely, if the variables differ at an $i$'th position, then either the word $\alpha_s^i \alpha_0 \alpha_s^{m-i}$ or $\alpha_s^i \alpha_1 \alpha_s^{m-i}$ does the job.
\end{proof}

\noindent\textbf{$\bullet$ Incrementation Gadget.}\\
This gadget is denoted by $V{\mathrel{++}}$, where $V$ is a variable.
It increases the value of the valid variable $V$ by $1$.
If the value of $V$ is the largest possible ($2^m-1$), then the gadget does not allow to obtain a proper configuration by any word of length $m+1$.
Variable $V$ must be valid in an initial configuration.

\begin{figure}[htb]\centering
\includegraphics[scale=1]{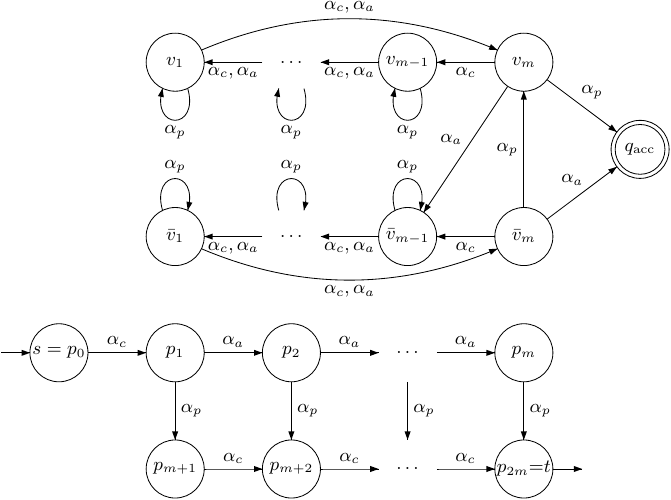}
\caption{Incrementation Gadget. All omitted transitions go to $q_{\rm acc}$.}\label{fig:incrementation_gadget}
\end{figure}

The gadget is illustrated in Fig.~\ref{fig:incrementation_gadget}.
It consists of control flow states $P = \{s=p_0,p_1,\ldots,p_{2m}=t\}$, one variable $V$, and letters $\Sigma = \{\alpha_a,\alpha_c,\alpha_p\}$.
The gadget performs a written addition of one to the value of $V$ interpreted in binary.
First, the letter $\alpha_c$ begins the process: this is the only letter that can be applied when $p_0$ is active.
Its transitions are defined as follows:
\begin{eqnarray*}
\delta(v_1,\alpha_c) & = & \{v_m\},\\
\delta(v_i,\alpha_c) & = & \{v_{i-1}\}\text{ for $i=2,\ldots,m$},\\
\delta(\bar{v}_1,\alpha_c) & = & \{\bar{v}_m\},\\
\delta(\bar{v}_i,\alpha_c) & = & \{\bar{v}_{i-1}\}\text{ for $i=2,\ldots,m$},\\
\delta(p_0,\alpha_c) & = & \{p_1\},\\
\delta(p_i,\alpha_c) & = & \{q_{\rm acc}\}\text{ for $i=1,\ldots,m$},\\
\delta(p_{m+i},\alpha_c) & = & \{p_{m+i+1}\}\text{ for $i=1,\ldots,m-1$},\\
\delta(p_{2m},\alpha_c) & = & \{q_{\rm acc}\}.
\end{eqnarray*}
Then, the letter $\alpha_a$ shifts the states of $V$ while also mapping ones ($v_m$) to zeros ($\bar{v}_{m-1}$), and it must be applied until the $m$'th position in $V$ is empty.
Every time when it is applied, the $m$'th position is cleared.
Its transitions are defined as follows:
\begin{eqnarray*}
\delta(v_1,\alpha_a) & = & \{v_m\},\\
\delta(v_i,\alpha_a) & = & \{v_{i-1}\}\text{ for $i=2,\ldots,m-1$},\\
\delta(v_m,\alpha_a) & = & \{\bar{v}_{m-1}\},\\
\delta(\bar{v}_1,\alpha_a) & = & \{\bar{v}_m\},\\
\delta(\bar{v}_i,\alpha_a) & = & \{\bar{v}_{i-1}\}\text{ for $i=2,\ldots,m-1$},\\
\delta(\bar{v}_m,\alpha_a) & = & \{ q_{\rm acc} \},\\
\delta(p_0,\alpha_a) & = & \{q_{\rm acc}\},\\
\delta(p_i,\alpha_a) & = & \{p_{i+1}\}\text{ for $i=1,\ldots,m-1$}.\\
\delta(p_{m+i},\alpha_a) & = & \{q_{\rm acc}\}\text{ for $i=1,\ldots,m$}.
\end{eqnarray*}
Then, the letter $\alpha_p$ must be applied, which fills the $m$'th position in $V$, and the active state $p_i$ is moved to the corresponding state $p_{m+i}$.
Its transition are as follows:
\begin{eqnarray*}
\delta(v_m,\alpha_p) & = & \{q_{\rm acc}\},\\
\delta(\bar{v}_m,\alpha_p) & = & \{v_m\},\\
\delta(v_i,\alpha_p) & = & v_i\text{ for $i=1,\ldots,m-1$},\\
\delta(\bar{v}_i,\alpha_p) & = & \bar{v}_i\text{ for $i=1,\ldots,m-1$},\\
\delta(p_0,\alpha_p) & = & \{q_{\rm acc}\},\\
\delta(p_i,\alpha_p) & = & \{p_{m+i}\}\text{ for $i=1,\ldots,m$},\\
\delta(p_{m+i},\alpha_p) & = & \{q_{\rm acc}\}\text{ for $i=0,\ldots,m$}.
\end{eqnarray*}
Finally, the letter $\alpha_c$ finishes the cyclic shifting of $V$, while moving the active control flow state over $p_{m+i},\ldots,p_{2m}$.

\begin{lemma}\label{lem:incrementation_gadget}
Let $C$ be an initial configuration with valid variable $V$ for the Incrementation Gadget $V{\mathrel{++}}$.
If $V(C) < 2^m-1$, then there exists a complete computation in $\Sigma^{m+1}$ from $C$ to a proper configuration $C'$.
Moreover, every complete computation has length $m+1$ and is such that $V(C')=V(C)+1$.
Longer computations are not proper, and if $V(C) = 2^m-1$, then every computation of length at least $m+1$ is not proper.
\end{lemma}
\begin{proof}
Consider a word $w$ of length $m+1$ and suppose that the obtained configuration $C'$ is proper.
So $t \in C'$ because of moving the active state in $P$ from $s$ to $t$.
Then $w$ must have the following form:
$$w = \alpha_c \alpha_a^i \alpha_p \alpha_c^{m-1-i},$$
for some $i \in \{0,\ldots,m-1\}$.
Consider the $j$'th occurrence of $\alpha_a$ ($1 \le j \le i$). 
If after reading it the current configuration is still proper, then it must be that $\bar{v}_j \notin C$ and so $v_j \in C$, since $v_j$ and $\bar{v}_j$ were mapped to $v_m$ and $\bar{v}_m$ by $\alpha_c \alpha_a^{j-1}$.
Next, it must be that $v_i \notin C$, because of the application of $\alpha_p$.
Since $w$ cyclically shifts $V$ exactly $m$ times, we end up with $C'$ such that:
\begin{itemize}
\item $v_j \in C'$ and $\bar{v}_j \notin C'$ for $j < i$,
\item $v_i \in C'$,
\item $v_j \in C'$ if and only if $v_j \in C$, and $\bar{v}_j \in C'$ if and only if $\bar{v}_j \in C$, for $j > i$.
\end{itemize}
Thus the value $V(C')$ is $V(C) - 2^1 - \dots - 2^{i-1} + 2^i = V(C)+1$.

If $V(C) < 2^m-1$, then there exists a smallest $i \le m$ such that
$v_i \notin C$, and there exists the unique word $w$ of that form which does the job.
\end{proof}

\noindent\textbf{$\bullet$ Assignment Gadget.}\\
This gadget is denoted either by $U \gets c$ or by $U \gets V$, where $c \in \{0,\ldots,2^m-1\}$ and $U$ and $V$ are two distinct variables.
It assigns to $U$ either the fixed constant $c$ or the value of the other variable $V$.
Variable $V$ must be valid in an initial configuration, but $U$ does not have to be.

The gadget consists of two control flow states $P = \{s,t\}$, a variable $U$ or two variables $U,V$, and the unary alphabet $\Sigma = \{\alpha\}$.
The transitions of $\alpha$ map $s$ to $t$, and additionally map either $s$ to the states of $U$ encoding value $c$ or the states of $V$ to the corresponding states of $U$.
The transitions are defined as follows:
\begin{eqnarray*}
\delta(u_i,\alpha) = \delta(\bar{u}_i,\alpha) & = & \emptyset,\\
\delta(t,\alpha) & = & \{q_{\rm acc}\}.
\end{eqnarray*}
If it assigns a fixed value $c$ then:
\begin{eqnarray*}
\delta(s,\alpha) & = & \{t\} \cup \{v_i \mid \text{$i$'th least bit of $c$ in binary is $1$}\} \\
                 &   & \cup\ \{\bar{v}_i \mid \text{$i$'th least bit of $c$ in binary is $0$}\}.
\end{eqnarray*}
If it assigns the value of $V$ then:
\begin{eqnarray*}
\delta(s,\alpha) & = & \{t\},\\
\delta(v_i,\alpha) & = & \{u_i,v_i\},\\
\delta(\bar{v}_i,\alpha) & = & \{\bar{u}_i,\bar{v}_i\}.
\end{eqnarray*}

In fact, the case $U \gets V$ could be alternatively implemented by a Selection Gadget followed by an Equality Gadget, although it will add more states and letters.

\noindent\textbf{$\bullet$ Waiting Gadget.}\\
This gadget is denoted by $\Call{Wait}{}_D$, where $D$ is a fixed positive integer.
This is a very simple gadget which just does nothing for $D$ number of letters.
There is exactly one complete computation, which has length $D$.

The gadget consists of states $Q = \{s=p_0,p_1,\ldots,p_D=t\}$ and the unary alphabet $\Sigma = \{\alpha\}$.
The transitions of $\alpha$ are defined as follows:
\begin{eqnarray*}
\delta(p_i,\alpha) & = & \{p_{i+1}\} \text{ for $i=0,\ldots,D-1$,}\\
\delta(p_D,\alpha) & = & \{q_{\rm acc}\}.
\end{eqnarray*}

\subsubsection{Joining gadgets together}

Compound gadgets are defined by other gadgets, which are joined together in the way specified by a program.
The method of definition is the only difference, as compound gadgets are objects of the same type as basic gadgets.
They will be also used incrementally to define further compound gadgets.

The general scheme for creating a compound gadget by joining gadgets $G_1,\ldots,G_k$ operating on variables from the sets $\mathcal{V}_1,\ldots,\mathcal{V}_k$ (all of width $m$), respectively, is as follows:
\begin{enumerate}
\item There are fresh (unique) control flow states of the gadgets, and there are the variables from $\mathcal{V}_1 \cup \cdots \cup \mathcal{V}_k$.
Thus when gadgets operate on the same variable, its states are shared.
\item The alphabet contains fresh (unique) copies of the letters of the gadgets.
\item Final states in the gadgets are also final in the compound gadget.
\item The transitions are defined as in the gadgets, whereas the transitions of a letter from a gadget $G_i$ map every control flow state that does not belong to $G_i$ to $\{q_{\rm acc}\}$ and fix the states of the variables on which $G_i$ does not operate.
\item Particular definitions of compound gadgets may additionally identify some of the start and target states of the gadgets and may add more (fresh) control flow states and letters.
\end{enumerate}

In our constructions, the control flow states with their transitions will form a directed graph, where the out-degree at every state of every letter is one (except the Parallel Gadget, defined later, which is an exception from the above scheme) -- it either maps a control flow state to another one or to $q_{\rm acc}$.
States of variables will never be mapped to control flow states.
This will ensure that during every proper computation from an initial configuration, exactly one control flow state is active.
The active control flow state will determine which letters can be used by a proper computation, i.e., the letters from the gadget owning this state (but in which it is not the target state).

Moreover, we will ensure that whenever a proper computation activates the start state of an internal gadget $G_i$, the current configuration restricted to the states of $G_i$ is an initial configuration for $G_i$ -- this boils down to assuring that the variables required to be valid have been already initialized (e.g., by a Selection Gadget or an Assignment Gadget).
Hence, complete computations for the compound gadget will contain complete computations for the internal gadgets, and the semantic properties of the compound gadget are defined in a natural way from the properties of the internal gadgets.

Now we define basic ways to join gadgets together.
Let $G_1,\ldots,G_k$ be some gadgets with start states $s^{G_1},\ldots,s^{G_k}$ and target states $t^{G_1},\ldots,t^{G_k}$, respectively.

\medskip
\noindent\textbf{$\bullet$ Sequence Gadget.}\\
For each $i=1,\ldots,k-1$, we identify the target state $t^{G_i}$ with the start state $s^{G_{i+1}}$.
Then $s^{G_1}$ and $t^{G_k}$ are respectively the start and target states of the Sequence Gadget.
We represent this construction by writing {\small $I_1\ \ldots\ I_k$}.
Complete computations for this gadget are concatenations of complete computations for the internal gadgets.

\begin{figure}[htb]\centering
\includegraphics[scale=1]{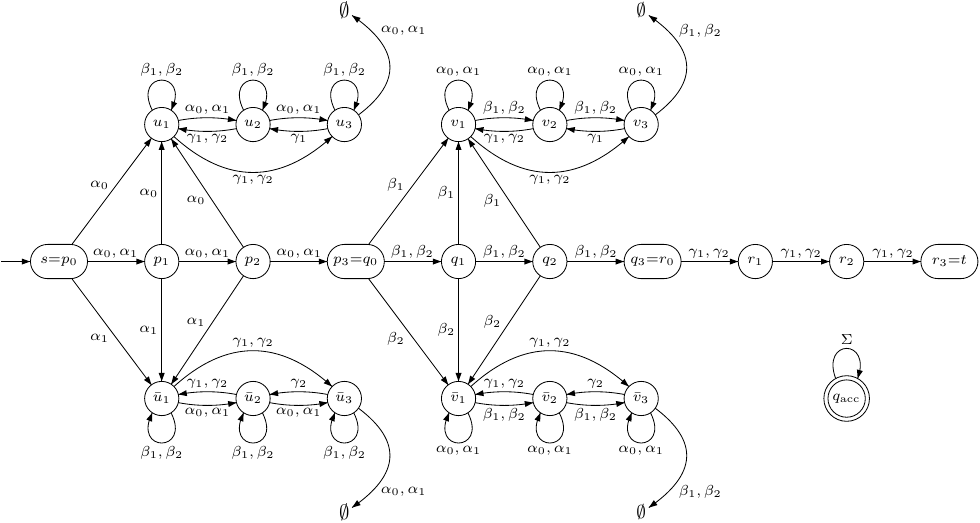}
\caption{The complete NFA of the Sequence Gadget {\small $\textsc{select}(U)\ \textsc{select}(V)\ U=V$}.
All omitted transitions go to $q_{\rm acc}$.
The states $p_i,q_i,r_i$ and letters $\alpha_i,\beta_i,\gamma_i$ belong to the three gadgets, respectively, that is:
$p_i=p_i^{\textsc{select}(U)}$, $\alpha_i=\alpha_i^{\textsc{select}(U)}$,
$q_i=p_i^{\textsc{select}(V)}$, $\beta_i=\alpha_i^{\textsc{select}(V)}$,
$r_i=p_i^{U=V}$, $\gamma_i=\alpha_i^{U=V}$.
}\label{fig:sequence_gadget_example}
\end{figure}

For example, for $k=3$ and $m=3$, the Sequence Gadget {\small $\Call{Select}{U}\ \Call{Select}{V}\ U=V$} is shown in Fig.~\ref{fig:sequence_gadget_example}.
It has the property that every complete computation has length $3m=9$ and is a concatenation of complete computations for the three gadgets; a final configuration $C'$ is such that $U(C')=V(C')$.
There exists a complete computation for every possible value of both variables, and longer computations are not proper.

\medskip
\noindent\textbf{$\bullet$ Choose Gadget.}\\
This gadget allows selecting one of the given gadgets nondeterministically.
We add a fresh start state $s$ and $k$ unique letters $\alpha_1,\ldots,\alpha_k$.
The action of a letter $\alpha_i$ maps $s$ to $\{s^{G_i}\}$, maps control flow states from the gadgets to $\{q_{\rm acc}\}$, and fixes variables states.
All target states $t^{G_i}$ are identified into the target state $t$ of the Choose Gadget.
We represent this construction by:
\begin{algorithmic}[0]\small
\Choose
\State $I_1$
\OrChoose
\State $\ldots$
\OrChoose
\State $I_k$
\EndChoose
\end{algorithmic}

Note that for this gadget there may exist complete computations of different lengths, even for the same initial configuration.
Nevertheless, there exists an upper bound on the length such that every computation longer than this bound is not proper (which is $1$ plus the maximum from the bounds for the internal gadgets).

\medskip
\noindent\textbf{$\bullet$ If--Else Gadget.}\\
It joins three gadgets, where $G_1$ is either an Equality Gadget or an Inequality Gadget, and $G_2$ and $G_3$ are any gadgets.
This construction is represented by:
\begin{algorithmic}[0]\small
\LineIfElse{$G_1$}{$G_2$}{$G_3$}
\end{algorithmic}
The third gadget $G_3$ may be empty, and then we omit the \emph{else} part.
The gadget is implemented as follows:
\begin{algorithmic}[0]\small
\Choose
\State $G_1$
\State $G_2$
\OrChoose
\State $\lnot G_1$
\State $G_3$
\EndChoose
\end{algorithmic}
where by $\lnot G_1$ we represent the negated version of $G_1$, i.e., the corresponding Inequality Gadget if $G_1$ is an Equality Gadget, and the corresponding Equality Gadget if $G_1$ is an Inequality Gadget.
Thus, a complete computation contains first a nondeterministic guess whether the variables in $G_1$ are equal, which is then verified, and then there is a complete computation for one of the two gadgets.

\medskip
\noindent\textbf{$\bullet$ While Gadget.}\\
This is easily constructed using \emph{If--Else}, where $G_3$ is empty, and the target state of $G_2$ is identified with the start state of the If--Else Gadget, which also becomes the start state of the While Gadget.
The target state $t$ of the While Gadget is the target state of {\small ${\lnot G_1}$}.

\medskip
\noindent{\textbf{$\bullet$ While-True Gadget.}\\
A special variant of the previous one is with $\mathbf{true}$ condition instead of $G_1$.
Therefore, it degenerates to the gadget $G_2$ with the target and the start states identified.
This is the only gadget that has outgoing transitions from its target state.
In contrast to the other constructions, there may exist infinitely many complete computations for this gadget (which are concatenations of shorter complete computations -- iterations of the loop).
We will be using it only as the last gadget in a Sequence Gadget, so there will be no possibility to leave this gadget.

\subsubsection{Arithmetic gadgets}

Now, we define gadgets for performing arithmetic operations by writing suitable programs, which represent these gadgets.
In each program we also define \emph{external} and \emph{internal} variables.
Internal variables are always fresh and unique and only this gadget operates on them, whereas external variables may be shared.

\medskip
\noindent{$\bullet$ \textbf{Addition Gadget.}}\\
The addition of two variables is denoted by $W \gets U+V$ and defined by Alg.~\ref{alg:addition}.
It is basically a Sequence Gadget with two Assignment Gadgets and a While Gadget.
It operates on four different variables.
Variables $U$ and $V$ must be valid in an initial configuration, while $W$ and $X$ do not have to be.

The semantic properties are such that, for an initial configuration $C$, if $U(C)+V(C) < 2^m-1$ then there exists a complete computation to a final configuration $C'$ such that $W(C')=U(C)+V(C)$.
If the result $U(C)+V(C)$ is larger than $2^m-1$, then the gadget does not admit a complete computation, because at some point we would have to go through the Incrementation Gadget in line~5 when the value of $U$ is equal to $2^m-1$.
Furthermore, every computation long enough is not proper.

\begin{algorithm}\caption{Addition Gadget $W \gets U+V$.}\label{alg:addition}
\begin{algorithmic}[1]\small
\item[\textbf{External variables:}]{$W,U,V$}
\item[\textbf{Internal variable:}]{$X$}
\State{$W \gets U$}\Comment{Assignment Gadget}
\State{$X \gets 0$}\Comment{Assignment Gadget}
\While{$X \neq V$}\Comment{While Gadget with Inequality Gadget}
  \State{$X{\mathrel{++}}$}\Comment{Incrementation Gadget}
  \State{$W{\mathrel{++}}$}\Comment{Incrementation Gadget}
\EndWhile
\end{algorithmic}
\end{algorithm}

\medskip
\noindent{\textbf{$\bullet$ Multiplication Gadget.}}\\
The multiplication of two variables is denoted by $W \gets U\cdot V$ and defined by Alg.~\ref{alg:multiplication}.
It uses an Addition Gadget to add $V(C)$ times the value $U(C)$ to the output variable $W$, where $C$ is an initial configuration.
As for Addition Gadgets, if the result $U(C)\cdot V(C)$ is larger than $2^m-1$, then this gadget does not admit a complete computation.

\begin{algorithm}\caption{Multiplication Gadget $U \cdot V$.}\label{alg:multiplication}
\begin{algorithmic}[1]\small
\item[\textbf{External variables:}]{$U,V,W$}
\item[\textbf{Internal variables:}]{$X,W'$}
\State{$W \gets 0$}
\State{$X \gets 0$}
\While{$X \neq V$}
  \State{$X{\mathrel{++}}$}
  \State{$W' \gets W+U$}\Comment{Addition Gadget}
  \State{$W \gets W'$}
\EndWhile
\end{algorithmic}
\end{algorithm}

\medskip
\noindent{\textbf{$\bullet$ Primality Gadget.}}\\
A primality test of the value of a variable $P$ is defined by Alg.~\ref{alg:primality}.
A complete computation is possible if and only if $P$ is prime.
We test whether $P$ is prime by enumerating all pairs of integers $2 \le X,Y < P$ and checking whether $X \cdot Y = P$.

\begin{algorithm}\caption{Primality Gadget.}\label{alg:primality}
\begin{algorithmic}[1]\small
\item[\textbf{External variable:}]{$P$}
\item[\textbf{Internal variables:}]{$X,Y,Z$}
\State{$X \gets 2$}
\While{$X \neq P$}
  \State{$Y \gets 2$}
  \While{$Y \neq P$}
    \State{$Y{\mathrel{++}}$}
    \State{$Z \gets X \cdot Y$}
    \State{$Z \neq P$}
  \EndWhile
  \State{$X{\mathrel{++}}$}
\EndWhile
\end{algorithmic}
\end{algorithm}

We will also need the negated version of this gadget (testing if $P$ is not prime), which can be implemented by selecting arbitrary values for two integers $X,Y$, and verifying that the values are not from $\{0,1,P\}$ and that $X\cdot Y = P$.
From now, a Primality Gadget or its negated version can be also used in an If-Else or a While gadget.

\medskip
\noindent{\textbf{$\bullet$ Prime Number Gadget.}}\\
Given two variables $U$ and $P$, this gadget assigns the $U(C)$'th prime number to $P$, where $C$ is an initial configuration.
It enumerates all integers $P=2,3,4,\ldots$, checks which are prime, and counts the prime ones in a separate variable $X$.
When $P$ becomes the $U(C)$'th prime number, the computation ends.
The gadget does not admit a complete computation when the $U(C)$'th prime number exceeds $2^m-1$, or when $U(C)=0$.

\begin{algorithm}\caption{Prime Number Gadget $P \gets U\text{'th prime number}$.}\label{alg:prime_number}
\begin{algorithmic}[1]\small
\item[\textbf{External variables:}]{$P,U$}
\item[\textbf{Internal variable:}]{$X$}
\State{$P \gets 2$}
\State{$X \gets 1$}
\While{$X \neq U$}
  \State{$P{\mathrel{++}}$}
  \LineIf{$P$ is prime}{$X{\mathrel{++}}$}
\EndWhile
\end{algorithmic}
\end{algorithm}

\medskip

The number of states in our NFAs of the gadgets is in $\O(dm)$, where $d$ is the length of the program measured by the number of basic gadgets instantiated plus the number of variables, and its alphabet has size $\varTheta(d)$.
For a given program (and an integer $m$), we can easily construct the corresponding NFA in polynomial time.

For each of the defined gadgets so far, except for While Gadget in general, there always exists an upper bound on the length of every complete computation, and every longer computation is not proper.
This bound is at most exponential in the size of the gadget, i.e., $\O(2^{dm})$, because proper computations cannot repeat the same configuration.

\subsection{An NFA with a large minimal universality length}\label{subsec:NFA_lower_bound}

Our first application is to show a lower bound on the maximum minimal universality length.
Its idea will be further extended to show the lower bound on the computational complexity.

\begin{algorithm}\caption{Large minimal universality length.}\label{alg:long}
\begin{algorithmic}[1]\small
\item[\textbf{Variables:}]{$X,Y$}
\State{\Call{Select}{$Y$}}\Comment{$[m]$}
\State{$X \gets 0$}\Comment{[$1$]}
\While{$\mathbf{true}$}
  \Choose\Comment{$[1]$}
    \State{$X = Y$}\Comment{$[m]$}
    \State{$X \gets 0$}\Comment{[$1$]}
    \State{[start final state] $\Call{Wait}{}_{m+1}$}\Comment{[$m+1$]}
  \OrChoose
    \State{$X \ne Y$}\Comment{$[m+1]$}
    \State{$X{\mathrel{++}}$}\Comment{[$m+1$]}
  \EndChoose
\EndWhile
\end{algorithmic}
\end{algorithm}

\begin{figure}[htb]\centering
\includegraphics[scale=1]{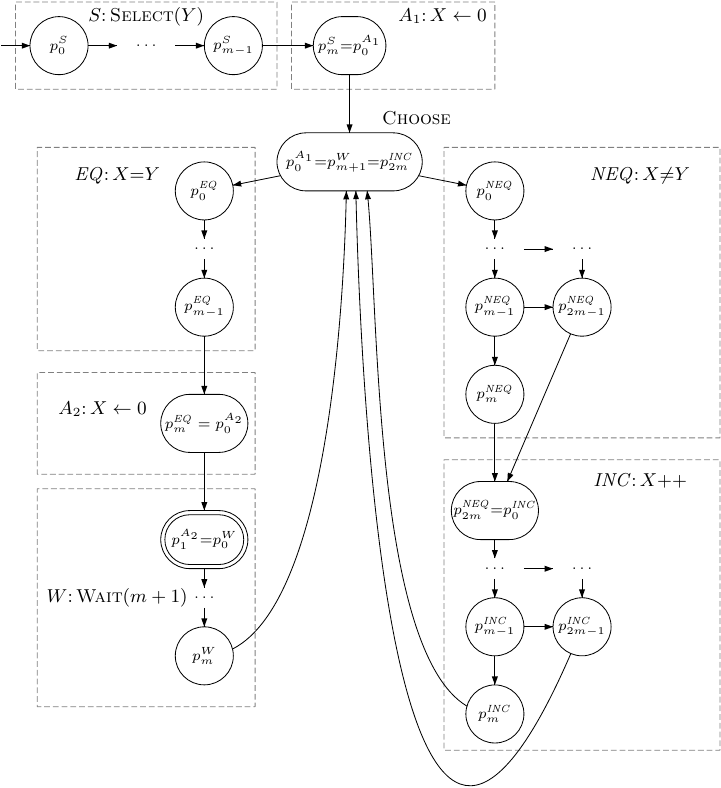}
\caption{The control flow states in the NFA of Alg.~\ref{alg:long} with their membership in particular gadgets.
Omitted transitions go to $q_{\rm acc}$.}\label{fig:alg_long}
\end{figure}

Alg.~\ref{alg:long} gives the program encoding our NFA, and Fig.~\ref{fig:alg_long} shows the control flow states of this NFA.
The numbers in the brackets [\, ] at the right denote the length of complete computations for the gadget (or for a part of it) in the current line.
In line~7, the annotation [start final state] indicates that the start control flow state of this Waiting Gadget is final, so the NFA of the program has two final states in total.

The idea of the program is as follows:
In the beginning, we select an arbitrary value for $Y$, and then in an infinite loop, we increment $X$ modulo $Y+1$.
The Choose Gadget in lines~4--11, is in fact, an If-Else Gadget with condition $X = Y$, unrolled for easier calculation of computation lengths.
Every iteration (complete computation of the Choose Gadget) of the loop takes the same number of letters ($2m+3$), hence given a computation length we know that we must perform $d-1$ complete iterations and end in the $d$'th iteration, for some $d$.
A proper computation of this length can avoid the final state in line~7 only in the iterations where the value of $X$ does not equal the value of $Y$.
We can ensure this for every length smaller than $\lcm(1,2,\ldots,2^m)\cdot(2m+3)$, as we can always select $Y$ such that $Y+1$ does not divide $d+1$, but it is not possible for length $\lcm(1,2,\ldots,2^m)\cdot(2m+3)$.

\medskip
\noindent\textbf{The analysis of the NFA of Alg.~\ref{alg:long}.}
Now, we have a very technical part, which is to calculate the size of the NFA of Alg.~\ref{alg:long} and ensure that the lengths are correct (see also Fig.~\ref{fig:alg_long}).

First, in line~1 any value for the variable $Y$ can be chosen, and this takes exactly $m$ letters of the Selection Gadget.
In line~2 the Assignment Gadget takes $1$ letter.
The While-True Gadget only means that the start and target states of the internal Choose Gadget are identified.
In line~4 a nondeterministic branch either to line~5 or~9 is performed, which takes $1$ letter.
Line~5 contains an Equality Gadget, which takes $m$ letters.
In line~6 we have the second Assignment Gadget, which takes $1$ letter.
In line~7 there is a Waiting Gadget that takes $m+1$ letters; there is also indicated that the start control flow state, which is also the target state of the second Assignment Gadget, is final.
In line~9 the Inequality Gadget takes $m+1$ letters, and the Incrementation Gadget in line~10 also takes $m+1$ letters.
Summarizing, each complete iteration of the while loop (thus a complete computation of the Choose Gadget) takes exactly $2m+3$ letters, regardless of the nondeterministic choice.

We count the number of states and the number of letters in the NFA.
We have two variables with $2m$ states each.
The Selection Gadget in line~1 has $m$ states (excluding the target state) and it has $2$ letters.
The Assignment Gadget in line~2 adds just $1$ control flow state and $1$ letter.
The while loop does not introduce any new states nor letters, but just the target and the start states of the internal Choose Gadget are identified.
In line~4 the Choose Gadget adds $1$ state and $2$ letters for branching.
In line~5 the Equality Gadget adds $m$ control flow states and $2$ letters.
In line~6 the Assignment Gadget adds $1$ state and $1$ letter.
In line~7 the Waiting Gadget adds $m+1$ states and $1$ letter.
In line~9 the Inequality Gadget adds $2m$ control flow states and $3$ letters.
Also, in line~10 the Incrementation Gadget adds $2m$ control flow states and $3$ letters.
The target state of the Waiting Gadget and of the Incrementation Gadget were identified with the start state of the Choose Gadget, thus we have already counted it.
Finally, there is the state $q_{\rm acc}$.
Summarizing, the NFA has $4m+m+1+1+m+1+(m+1)+2m+2m+1=11m+5$ states and $2+1+2+2+1+1+3+3=15$ letters.

Now, we prove that the NFA has a long minimal universality length.
\begin{lemma}\label{lem:long}
For a given $m$, the NFA of Alg.~\ref{alg:long} has minimal universality length
$$\lcm(1,2,\ldots,2^m)\cdot(2m+3).$$
The number of states of this NFA is $11m+5$ and the size of its alphabet is $15$.
\end{lemma}
\begin{proof}
There are two final states in the NFA: $q_{\rm acc}$ and the start state of the Waiting Gadget indicated in line~7.
Thus, every non-accepted word must be a proper computation and such that the current configuration does not contain this final state from line~7.

First, we show that there exists a non-accepted word for every length smaller than $\lcm(1,2,\ldots,2^m)\cdot(2m+3)$.
Observe that for every length there exists a word $w$ being a proper computation of the program.
Moreover, there exists such a word $w$ for every value of $Y$ selected in line~1 (if $|w|\ge m$).
Consider such a word $w$ for some value of $Y$.
From the beginning of the program to (a configuration with) the final state a proper computation takes at least $2m+3$ letters, so if $|w| < 2m+3$ then $w$ is not accepted.
Since every iteration of the while loop takes exactly $2m+3$ letters, if $|w|$ is not divisible by $2m+3$, then the final state in line~7 cannot be active after reading the whole $w$.
Suppose that $|w|$ of length at least $2m+3$ is divisible by $2m+3$ and let $d = |w|/(2m+3) \ge 1$.
Then the computation of $w$ performs exactly $d-1$ iterations of the while loop and ends in the $d$'th iteration.
Depending on the value of $Y$, after reading $w$ the active control flow state may be either the final state in line~7 or the start state of the Incrementation Gadget in line~10.
If $|w| < \lcm(1,2,\ldots,2^m)\cdot(2m+3)$, then $d < \lcm(1,2,\ldots,2^m)$.
So we can find a value for $Y$ ($1 \le Y \le 2^m-1$) such that $d$ is not divisible by $Y+1$.
The computation $w$ is such that there are $Y$ iterations in which $X$ is incremented and one iteration in which $X$ is reset to $0$, which is repeated during the whole computation.
So at the beginning of an $i$'th iteration the value of $X$ is equal to $(i-1) \bmod (Y+1)$.
Since $d$ is not divisible by $Y+1$, at the beginning of the $d$'th iteration we have $X=(d-1) \bmod (Y+1) \neq Y$, and so the computation finishes at the first state of the Incrementation Gadget in line~10.

It remains to show that every word $w$ of length $\lcm(1,2,\ldots,2^m)\cdot(2m+3)$ is accepted.
Suppose, contrary to what we want, that $w$ is not accepted, which means by definition that it must be a proper computation.
Hence, in the beginning, it chooses some value for $Y$.
Then it performs iterations in the while loop.
Since every iteration takes exactly $2m+3$ letters, exactly $\lcm(1,2,\ldots,2^m)-1$ complete iterations must be performed, and the computation ends in the $\lcm(1,2,\ldots,2^m)$'th iteration.
Regardless of the selected value for $Y$, $\lcm(1,2,\ldots,2^m)$ is divisible by $Y+1$.
Hence, at the beginning of the last iteration, we have $X=Y$, and there remain $m+2$ letters to read.
Now, if $w$ chooses the second branch, the computation cannot pass the test in line~9, since for this Inequality Gadget there is no complete computation and every computation longer than $m+1$ is not proper.
So $w$ must choose the first branch, which after $m+2$ letters results in a configuration with the final state in line~7.
\end{proof}

It remains to calculate the bound in terms of the number of states.
\begin{theorem}\label{thm:lower_bound_nfa}
For a $15$-letter alphabet, the minimal universality length can be as large as
$$e^{2^{n/11} (1+o(1))}.$$
\end{theorem}
\begin{proof}
From the prime number theorem, we have
$$\lcm(1,2,\ldots,2^m)\cdot(2m+3) \in (2m+3) \cdot \exp(2^m (1+o(1)) = \exp(2^m (1+o(1))).$$
For a sufficiently large number of states $n$ we can construct the NFA from Alg.~\ref{alg:long} for $m = \lfloor (n-5)/11\rfloor$, and add $n-m$ unused states.
Then we obtain
$$\exp(2^m (1+o(1))) = \exp(2^{(n-5)/11} (1+o(1))) = \exp(2^{n/11} (1+o(1))).$$
\end{proof}

\begin{remark}
The constants in Lemma~\ref{lem:long} and Theorem~\ref{thm:lower_bound_nfa} are not optimal and could be further optimized.
The witness NFA has been obtained directly from the construction, whereas, for example, some control flow states could be manually shared and the number of letters reduced.
\end{remark}

\subsection{Controlling the computation length}\label{subsec:controlling_length}

As we noted, because of Choose Gadgets, complete computations may have different lengths.
For example, the Addition Gadget for an initial configuration $C$ performs $V(C)$ iterations of its internal while loop.
Moreover, two or more branches of a Choose Gadget may admit complete computations of different lengths even for the same current configuration.
This is an obstacle that makes it difficult or impossible to further rely on the exact length $|w|$ of a proper computation, based on which we would like to decide if $w$ must be accepted.
Therefore, if we want to still use our constructions, we need a possibility to ensure that all complete computations have a fixed known length, and furthermore, that there are no proper computations longer than that length.

\medskip
\noindent\textbf{$\bullet$ Delaying Gadget.}\\
The first new ingredient is the Delaying Gadget $\Call{Delay}{}_D$, where $D$ is a fixed integer $\ge 0$.
This is a stronger version of the Waiting Gadget.
It has the advantage that it can wait for an exponential number of letters in the size of the gadget, although by the cost of that not all computation lengths are possible to encode.

Let $D \ge 1$ be a fixed integer and let $m' \ge 1$ be an integer such that $D \le 2^{m'}-1$.
The program uses two variables $X$ and $Y$, both of width $m'$.
The width $m'$ can be different from $m$, which is used for all variables in all other gadgets, but this is allowed since the variables $X$ and $Y$ are always internal and will never be shared.
The Delaying Gadget is defined by Alg.~\ref{alg:delaying_gadget} and denoted by $\Call{Delay}{}_D$.

The length of every complete computation is precisely $2+D(1+2(m'+1))+1+m'$, and every longer computation is not proper.
Let $T(D)$ denote this exact length when we take $m'$ to be the minimum possible, so $m'=\lceil\log_2 (D+1)\rceil$.
Thus, $T(D) \in \varTheta(D\log D)$.

\begin{algorithm}\caption{Delaying Gadget.}\label{alg:delaying_gadget}
\begin{algorithmic}[1]\small
\item[\textbf{Internal variables:}]{$X,Y$ of width $m'$}
\State{$X \gets 0$}\Comment{[$1$]}
\State{$Y \gets D$}\Comment{[$1$]}
\While{$X \neq Y$}\Comment{[$1+(m'+1)$]}
  \State{$X{\mathrel{++}}$}\Comment{[$m'+1$]}
\EndWhile\Comment{[$m'$] (here is an Equality Gadget)}
\end{algorithmic}
\end{algorithm}

\medskip
\noindent\textbf{$\bullet$ Parallel Gadget.}\\
The idea to control the computation length is to implement computation in parallel.
A given gadget for which there may exist complete computations of different lengths is computed in parallel with a Delaying Gadget.
When the computation is completed for the given gadget, we still must wait in its target state until the computation for the Delaying Gadget is also finished.
In this way, as long as complete computations for the Delaying Gadget are always longer than those for the given gadget (we can ensure this by choosing $D$), complete computations for the joint construction will have fixed length $T(D)+1$.
The joint computation is realized by replacing the alphabet with new letters for every combined pair of actions in both gadgets.

We construct Parallel Gadget as follows.
Let $G_1$ be a given gadget for which there exists an upper bound $L$ such that every longer computation is not proper.
We assume that all outgoing transitions from its target state go to $q_{\rm acc}$ (only the While-True Gadget violates this).
The second gadget $G_2$ will be the Delaying Gadget with a $D$ such that $T(D) \ge L$.
Parallel Gadget is illustrated in Fig.~\ref{fig:parallel_gadget} and defined as follows.
\begin{itemize}
\item The start states $s^{G_1}$ and $s^{G_2}$ are identified with the start state $s$ of the Parallel Gadget.
The other states from both gadgets are separate.
The target state $t$ of the Parallel Gadget is a fresh state.
\item The alphabet consists of the following:
\begin{itemize}
\item For every pair of letters $\alpha^{G_1}_i$ from $G_1$ and $\alpha^{G_2}_j$ from $G_2$, there is the letter $\beta_{i,j}$ that acts on the states from $G_1$ as $\alpha^{G_1}_i$ and on the states from $G_2$ as $\alpha^{G_2}_j$.
\item For every letter $\alpha^{G_2}_j$ from $G_2$, there is the letter $\gamma_j$ that acts on the states from $G_2$ as $\alpha^{G_2}_j$, fixes $t^{G_1}$ and the variable states in $G_1$, and maps all the other control flow states of $G_1$ to $q_{\rm acc}$.
\item There is the letter $\tau$ that maps both $t^{G_1}$ and $t^{G_2}$ to $\{t\}$, fixes all variable states, and maps all the other control flow states to $q_{\rm acc}$.
\item All letters map $t$ to $q_{\rm acc}$.
\end{itemize}
\end{itemize}

\begin{figure}[htb]\centering
\includegraphics[scale=1]{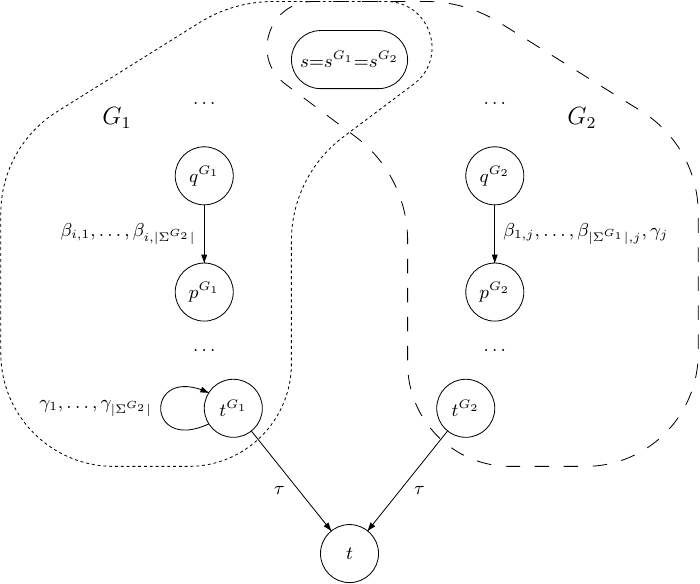}
\caption{The Parallel Gadget, where in $G_1$ a letter $\alpha^{G_1}_i$ maps a state $q^{G_1}$ to $p^{G_1}$, and in $G_2$ a letter $\alpha^{G_2}_j$ maps a state $q^{G_2}$ to $p^{G_2}$.}\label{fig:parallel_gadget}
\end{figure}

The Parallel Gadget works as follows:
First, letters $\beta_{i,j}$ must be used, which encodes proper computations for both gadgets.
When $t^{G_1}$ becomes active, the computation for the $G_1$ is completed and $\beta_{i,j}$ cannot be used anymore, but the letters $\gamma_i$ can be applied.
When letters $\gamma_i$ are applied, $G_1$ waits until the Delaying Gadget has finished the computation.
Finally, when and only when both $t^{G_1}$ and $t^{G_2}$ are active, we can and must apply $\tau$, which finishes a complete computation of the Parallel Gadget and the target state becomes active.
Therefore, the length of all complete computations of the Parallel Gadget is equal to $T(D)+1$, which is the length for the Delaying Gadget plus the letter $\tau$.

Note that we can construct the Parallel Gadget in polynomial time.
If $n$ is the number of states in $G_1$, we know that its complete computations cannot be longer than $2^n$ (so $L \le 2^n$).
Thus we can simply set $D=2^n$ and $m'=n+1$, which for sure is such that $T(D) > L$.
The number of states and letters of the Parallel Gadget is still polynomial in the size of $G_1$.

This gadget has a different nature than the previous ones because it admits that a proper computation yields a proper configuration with more than one control flow state active.

\subsection{Length divisibility}\label{subsec:divisibility}

In Subsection~\ref{subsec:NFA_lower_bound} we have constructed an NFA for which every word encoding a proper computation must be accepted or can be not accepted depending on its length, namely, it is always accepted if the length is divisible by $\lcm(1,2,\ldots,2^m)(2m+3)$.
We generalize this idea so that we will be able to express more complex properties about the length of which all words must be accepted.

We are going to test whether $|w|$ satisfies some properties, in particular, whether a function of $|w|$ is divisible by some integers.
This extends the idea from Alg.~\ref{alg:long}, which just verifies whether $|w|/r$, for some constant $r$, is not divisible by some integer from $2$ to $2^m$.

\begin{algorithm}\caption{Divisibility program.}\label{alg:divisibility}
\begin{algorithmic}[1]\small
\item[\textbf{Variables:}]{$X_1,\ldots,X_k,X'_1,\ldots,X'_k$}
\State{$\Call{Select}{X_1},\ldots,\Call{Select}{X_k}$.}\Comment{[$km$]}
\State{$X'_1 \gets 0,\ldots,X'_k \gets 0$}\Comment{$[k]$}
\While{$\mathbf{true}$}
  \Choose\Comment{[$1$]}
    \State{{\scriptsize\stackanchor{execute}{in parallel}} $\begin{cases}$
    $\Call{Delay}{}_D$ $ \\ $
    Verifying procedure
    $\end{cases}$}\Comment{[$T(D)+1$]}
  \OrChoose
    \State{$\Call{Delay}{}_D$}\Comment{[$T(D)$]}
    \State{[final state] $\Call{Wait}{}_1$}\Comment{[$1$]}
  \EndChoose
  \For{$i = 1,\ldots,k$}
    \State{$X'_i{\mathrel{++}}$}\Comment{[$m+1$]}
    \If{$X'_i = X_i$}\Comment{[$m+1$ if $X'_i=X_i$, and $m+2$ otherwise]}
      \State{$X'_i \gets 0$}\Comment{[$1$]}
    \EndIf
  \EndFor
\EndWhile
\end{algorithmic}
\end{algorithm}

We define the \emph{divisibility program} shown in Alg.~\ref{alg:divisibility}.
It is constructed for given numbers $k$ and $m$, and a \emph{verifying procedure}, which is any gadget satisfying the following properties:
\begin{itemize}
\item Every complete computation of the gadget is not longer than an $L \in 2^{\O(\mathrm{poly}(km))}$. All longer computations are not proper.
All outgoing transitions from the target state go to $q_{\rm acc}$.
These conditions will allow synchronizing the length of all complete computations to one length $T(D)+1 > L$ with a Parallel Gadget.
\item It may use, but does not modify the variables $X_1,\ldots,X_k,X'_1,\ldots,X'_k$, i.e., after its every complete computation their values in the final configuration are the same as in the initial configuration. Also, it may have internal variables, which do not have to be valid in an initial configuration, but the existence of complete computations cannot depend on their setting in an initial configuration of the gadget.
These conditions ensure that the gadget can be activated repetitively in the same proper computation of the whole program.
\item It does not contain final states.
\end{itemize}

As before, the numbers in the brackets [\, ] at the right denote the lengths of a word of a complete computation for the current line.
There is an infinite while loop, which consists of two parts.
In the first part, a nondeterministic choice is made (line~4): either to run the verifying procedure or to wait.
For the verifying procedure (line~5) we use the Parallel Gadget; this ensures that this part finishes after exactly $T(D)+1 > L$ letters.
In the waiting case (line~7), we use the Delaying Gadget with the same value of $D$ as that in the Parallel Gadget.
Then there is a single final state (line~8).
In the second part (lines~10--15), every variable $X'_i$ counts the number of iterations of the while loop modulo $X_i$.
The \emph{for} loop denotes that the body is instantiated for every $i$ (it is a Sequence Gadget).
Every complete computation of the second part (lines~10--15) has exactly $(2m+3)k$ letters.

The idea is that, for certain lengths, every proper computation must end with a configuration with the non-final control flow states in line~5 (these are precisely the two target states, of the verifying procedure and of the Delaying Gadget) or with the final state in line~8.
However, for the first option, it must succeed in the last iteration with the verifying procedure when $X'_i = \ell' \bmod X_i$.
In other words, for some selection of the values for $X_1,\ldots,X_k$, there must exist a complete computation of the verifying procedure from an initial configuration with these values for $X_i$ and $X'_i = \ell' \bmod X_i$.
Due to these auxiliary variables $X'_i$, the verifying procedure can check the divisibility of $\ell$ by $X_i$.
\begin{lemma}\label{lem:divisibility}
Consider Alg.~\ref{alg:divisibility} for some $k$, $m$, and a verifying procedure.
There exist integers $r_1 \ge 1$ and $r_2 \ge 1$ such that the NFA of Alg.~\ref{alg:divisibility} accepts all words of a length $\ell$ if and only if there exist an integer $\ell' \ge 0$ such that:
\begin{itemize}
\item $\ell = r_1 \cdot \ell' + r_2$, and
\item for every initial configuration $C$ of the verifying procedure where variables $X_1,\ldots,X_k$ are valid and $X'_i(C) = \ell' \bmod X_i(C)$ for all $1 \le i \le k$, there does not exist a complete computation for the verifying procedure.
\end{itemize}
\end{lemma}
\begin{proof}
Let $r_1=1+(T(D)+1)+(2m+3)k$, which is the length complete computations of one iteration of the while loop (lines~4--15), and let $r_2 = km+k+1+T(D)$, which is the length of a proper computation from a configuration with the initial state (start state in line~1) to a configuration with the final state in line~8 that does not contain a complete computation of the while loop.
These values depend only on the algorithm, so on $k$, $m$, and the verifying procedure.

Consider a length $\ell$ that is not expressible as $r_1 \cdot \ell' + r_2$.
There exists a proper computation $w$ of length $\ell$ that every time chooses the second branch (lines~7--8).
It is not accepted since the obtained configuration is proper and cannot contain the final state in line~8.

Consider a word $w$ of length $r_1 \cdot \ell' + r_2$.
If $w$ is not accepted, then it must encode a proper computation.
Then it must perform exactly $\ell'$ complete iterations of the while loop, and the last control flow state is either the final state in line~8 or the non-final states in line~5.
We know that at the beginning of the last (incomplete) iteration we have $X'_i = \ell' \bmod X_i$ for all $i$.
Now, if there exists a selection of the values for $X_1,\ldots,X_k$ such that there exists a complete computation of the verifying procedure for the initial configurations with these values of $X_i$ and $X'_i$, then $w$ can select these values of $X_i$ and choose the first branch in the last iteration. Thus there exists a proper computation $w$ that ends with a configuration with the non-final control states in line~5.
Otherwise, since in the last iteration entering the verifying procedure results in a non-proper configuration after $T(D)$ letters, regardless of the choice for $X_1,\ldots,X_k$, $w$ must choose the second branch in the last iteration of the while loop, which results in the final state in line~8 in the last configuration.
\end{proof}

\subsubsection{Existential divisibility formulas}

We develop a method for verifying the properties of the computation in a flexible way.
We use a subset of first-order logic, where formulas are in a special form.
For a given integer $m$, we say that a formula $\varphi$ is in \emph{existential divisibility} form if its only free variable is $\ell'$ (not necessarily occurring in $\varphi$) and it has the following form:
$$\exists_{X_1,\ldots,X_k \in \{0,\ldots,2^m-1\}}\; \psi(X_1,\ldots,X_k,\ell').$$
Formula $\psi$ is any propositional logic formula that uses operators $\land$, $\lor$, and whose simple propositions are of the following possible forms:
\begin{enumerate}
\item $(X_i = c)$, where $c \in \{0,\ldots,2^m-1\}$,
\item $(X_h = X_i + X_j)$,
\item $(X_h = X_i \cdot X_j)$,
\item $X_i$ is prime,
\item $X_i$ is the $X_j$'th prime number,
\item $(X_i \mid \ell')$ or $(X_i \nmid \ell')$,
\end{enumerate}
where $X_i,X_j,X_h$ are some variables from $\{X_1,\ldots,X_k\}$.

Given a $\varphi$, we can ask for what integer values of $\ell'$ the formula is satisfied, and in particular whether it is not a tautology over positive integers. 
\begin{problem}\label{pbm:ns-edf}
Given an existential divisibility formula $\varphi$, is there a positive integer $\ell'$ such that $\varphi(\ell')$ is not satisfied?
\end{problem}

\medskip
\noindent\textbf{$\bullet$ Verifying Gadget.}\\
Recall the formula $\psi$ occurring in an existential divisibility formula $\varphi$.
We construct the gadget $\Call{Verify}{\psi}$ for verifying $\psi$.
The gadget uses the external variables $X_1,\ldots,X_k$, which are assumed to correspond with those in $\psi$, and the external auxiliary variables $X'_1,\ldots,X'_k$.
There are also some fresh internal variables.
The value of $\ell'$ is not given, but instead, we assume that the value of every $X'_i$ is equal to $\ell' \bmod X_i$ (and $0$ when $X_i=0$), hence we will be able to check the divisibility of $\ell'$.

It is built using Sequence Gadgets for conjunctions, Choose Gadgets for disjunctions, and other appropriate gadgets for (1)--(6).
In more detail, $\Call{Verify}{\psi}$ is defined recursively as follows.
\begin{itemize}
\item For (1), when $\psi$ is $(X_i = c)$:\\
We add a fresh unique variable $C$.
First we use (in a Sequence Gadget) the Assignment Gadget $C \gets c$, and then we use the Equality Gadget $X_i = C$.
\item For (2) and (3), when $\psi$ is $(X_i + X_j = X_h)$ or $(X_i \cdot X_j = X_h)$:\\
We add a fresh unique variable $C$.
First we assign $C \gets X_i+X_j$ by the Addition Gadget or $C \gets X_i \cdot X_j$ by the Multiplication Gadget, and then we use the Equality Gadget $C = X_h$.
\item For (4), when $\psi$ is ($X_i$ is prime):\\
We use the Primality Gadget.
\item For (5), when $\psi$ is ($X_i$ is the $X_j$'th prime number):\\
We add a fresh unique variable $C$.
First we compute the $X_j$'th prime number using the Prime Number Gadget, and then we use the Equality Gadget.
\item For (6), when $\psi$ is $(X_i \mid \ell')$ or $(X_i \nmid \ell')$:
We use either the Equality Gadget $X'_i = 0$ or the Inequality Gadget $X'_i \neq 0$, which verify the divisibility under our assumption that $X'_i = \ell' \bmod X_i$.
\item When $\psi = \psi_1 \land \dots \land \psi_h$:\\
It is the Sequence Gadget joining the verifying gadgets for the subformulas,
i.e.,
\begin{algorithmic}[0]\small
\State $\Call{Verify}{\psi_1} \dots \Call{Verify}{\psi_h}$
\end{algorithmic}
\item When $\psi = \psi_1 \lor \dots \lor \psi_h$:\\
It is the Choose Gadget joining the verifying gadgets for the subformulas, i.e.,
\begin{algorithmic}[0]\small
\LineChooseThree{$\Call{Verify}{\psi_1}$}{\ldots}{$\Call{Verify}{\psi_h}$}
\end{algorithmic}
\end{itemize}
Note that for (2), (3), and~(5) if the value of the operation exceeds $2^m-1$, then for sure this simple proposition is false and a complete computation does not exist for the gadget computing this value.

Alg.~\ref{alg:example_verifying} shows the program describing $\Call{Verify}{\psi}$ for an example formula $\psi$.
\begin{algorithm}\caption{The Verifying Gadget for $k=3$ and
\[ \psi = (X_1 = X_2+X_3) \lor ((X_1\text{ is the }X_2\text{'th prime number}) \land (X_2 \nmid \ell')). \]
}\label{alg:example_verifying}
\begin{algorithmic}[1]\small
\item[\textbf{External variables:}] $X_1,\ldots,X_3,X'_1,\ldots,X'_3$
\item[\textbf{Internal variables:}] $C_1,C_2$
\Choose
\State $C_1 \gets X_2+X_3$
\State $X_1 = C_1$
\OrChoose
\State $C_2 \gets X_2\text{'th prime number}$
\State $X_1 = C_2$
\State $X'_2 \neq 0$
\EndChoose
\end{algorithmic}
\end{algorithm}

It follows that, for an initial configuration $C$ for \Call{Verify}{$\psi$} with valid variables $X_1,\ldots,X_k$ and where $X'_i(C) = \ell' \bmod X_i(C)$, there exists a complete computation if and only if $\psi(X_1,\ldots,X_k,\ell')$ is satisfied.
There exists an exponential upper bound $L$ (in the size of the gadget) on the length of all complete computations (although they may have different lengths due to disjunctions), and longer computations are always not proper.
Furthermore, the values of the variables $X_1,\ldots,X_k,X'_1,\ldots,X'_k$ remain the same in the final proper configuration.
The internal variables are always initialized, so they can be arbitrary in an initial configuration and complete computations do not depend on their setting.
Finally, there are no final states.
Hence, \Call{Verify}{$\psi$} satisfies the requirements for being the verifying procedure in Alg.~\ref{alg:divisibility}.

\subsubsection{Reduction from Problem~\ref{pbm:ns-edf}.}
We already have all ingredients to reduce Problem~\ref{pbm:ns-edf} to Problem~\ref{pbm:elu} (existential length universality).
Given an existential divisibility formula
$$\varphi(\ell') = \exists_{X_1,\ldots,X_k \in \{0,\ldots,2^m-1\}}\; \psi(X_1,\ldots,X_k),$$
we construct the program from Alg.~\ref{alg:divisibility} with the Verifying Gadget $\Call{Verify}{\psi}$ as the verifying procedure.
Hence, the formula is translated to an NFA in polynomial time.
By Lemma~\ref{lem:divisibility}, there exist integers $r_1,r_2 \ge 1$ such that the NFA accepts all words of some length $\ell$ if and only if for some integer $\ell' \ge 0$, $\ell=r_1\cdot \ell'+r_2$ and $\varphi(\ell')$ is not satisfied.
Hence, if $\varphi$ is not satisfied for some $\ell'$, then the NFA accepts all words of length $\ell$, and if the NFA accepts all words of a length $\ell$, then $\ell$ must be expressible as $r_1\cdot \ell'+r_2$ and $\varphi(\ell')$ must be not satisfied.

\begin{remark}
With a few more technical steps, which require, e.g., adding the negation and controlling variable bounds, it is possible to represent the negated Problem~\ref{pbm:ns-edf} as the satisfiability problem of the Presburger arithmetic with the prefix class $\exists\forall^*$ and whose formulas are of a specific form.
The Presburger arithmetic with the prefix class $\exists\forall^*$ is NEXPTIME-hard \cite{Gradel:1989}, and a little more general one with the prefix class $\exists^*\forall^*$ is $\Sigma_1^\mathrm{EXP}$-complete \cite{Haase:2014}.
However, our problem is a strict subclass of the first case, because the first and the only unbounded variable $\ell'$ can be checked only for divisibility, all the other variables are exponentially bounded, and the propositions are of particular forms.
Hence, we cannot directly infer the hardness from that known result.
In the last reduction step, we show that NEXPTIME-hardness still holds for our restricted problem.
\end{remark}

\subsection{Reduction to the non-satisfiability of existential divisibility formulas}

In the final step, we reduce from the canonical NEXPTIME-complete problem:
given a nondeterministic Turing machine $N$ with $s$ states, does it accept the empty input after at most $2^{s}$ steps?
Without loss of generality, $N$ is a one-tape machine using the binary alphabet $\{0,1\}$.
Furthermore, $N$ can be modified so that, upon accepting, it clears the tape, moves the head to the leftmost cell, and waits there while being still in the (unique) accepting state $q_f$.
Because we are interested in executing at most $2^{s}$ steps, we can also bound the length of the tape.
Therefore, we can assume that the tape of length $2^{s}$ is initially filled with $2^{s}$ zeroes, the head is on the leftmost cell and the machine is in the unique initial state $q_0$.
The goal is then to check if there exists a computation such that, after exactly $2^{s}$ steps, the head is on the leftmost cells and the machine is in the unique accepting state. We work with such a formulation from now on.

A computation of $N$ can be represented by a $2^{s}\times 2^{s}$ table, where the $i$'th row of the table describes the content of the tape after $i$ steps of the computation.
Every cell $(r,c)$ ($0 \le r,c \le 2^s-1$) of the table stores a symbol from $\{0,1\}$ and, possibly, a state of $N$.
Therefore, to check if there exists an accepting computation we need to check if it is possible to fill the table so that it represents subsequent steps of such a computation.

Let $Q$ be the set of states of $N$, and $q_0,q_f \in Q$ be its starting and accepting states, respectively.
We want to check if it is possible to choose an element $t(r,c)\in (Q\cup\{\mathrm{nil}\})\times\{0,1\}$ for every cell $(r,c)$ of an $2^{s}\times 2^{s}$ table, so that the whole table describes an accepting computation of $M$.
The elements are identified with integers by a function $f\colon (Q\cup\{\mathrm{nil}\})\times\{0,1\}) \rightarrow \{0,\ldots,z-1\}$, where $z = 2s+2$.
Therefore, we want to choose a number $t(r,c)$ from $\{0,\ldots,z-1\}$ for every $(r,c)$.
We encode all these choices in one non-negative integer $\ell'$ as follows.
Let $p(k)$ denote the $k$'th prime number.
For every $(r,c)$, we reserve $z$ prime numbers $p((2^{s}r+c)z+1),\ldots,p((2^{s}r+c)z+z)$.
Each of these primes represents a possible choice for $t(r,c)$.
Then, we select the remainder of $\ell'$ modulo $p((2^{s}r+c)z+i)$ to be zero if $t(r,c)=i$;
otherwise, we select any non-zero remainder.
Then, by the Chinese remainder theorem, any choice of all the elements can be represented by a non-negative integer $\ell' \leq p(1) \cdots p(2^{2s}z)$.
In the other direction, every non-negative integer $\ell'$ represents such a choice as long as, for all $(r,c)$, $\ell'$ is divisible by exactly one of the $z$ primes reserved for $(r,c)$.
Hence, we now focus on constructing an existential divisibility formula $\varphi(\ell')$ that can be used to check if indeed $\ell'$ has such a property and, if so, whether the represented choice describes an accepting computation of $N$.

We construct $\varphi(\ell')$ so that it is satisfied exactly when at least one of the following situations occur:
\begin{enumerate}
\item \label{error:twice} For some $(r,c)$ and $1\leq i < j \leq z$, $\ell'$ is divisible by both $p((2^{s}r+c)z+i)$ and $p((2^{s}r+c)z+j)$.
\item For some $(r,c)$, $\ell'$ is not divisible by $p((2^{s}r+c)z+i)$, for every $i=1,\ldots,z$.
\item $\ell'$ is not divisible by $p(f(q_0,0))$.
\item For some $c \in \{1,\ldots,2^{s}-1\}$, $\ell'$ is not divisible by $p(c\cdot z+f(\mathrm{nil},0))$.
\item $\ell'$ is not divisible by $p(2^{s}(2^{s}-1)z+f(q_f,0))$.
\item \label{error:legal} For some $r \in \{1,\ldots,2^s-1\}$ and $c \in \{2,\ldots,2^s-1\}$, the $2\times 3$ window of cells with the lower-right corner at $(r,c)$ is not legal for the transitions of $N$.
\end{enumerate}
Before describing how to construct such a formula, we elaborate on the last condition.
A $2\times 3$ window of cells is legal if it is consistent with the transition function of $N$.
We avoid giving a tedious precise definition and only specify that $W\subseteq \{0,\ldots,z-1\}^{6}$ is the set of six-tuples of numbers corresponding to elements chosen for the cells of such a legal $2\times 3$ window;
$W$ can be constructed in polynomial time given the transition function of $N$.
Therefore, the last condition can be written in more detail as follows.
\begin{enumerate}
\item[\ref{error:legal}'.] For some $r \in \{1,\ldots,2^s-1\}$ and $c \in \{2,\ldots,2^s-1\}$ and $(i_{1,1},i_{1,2},\ldots,i_{2,3})\notin W$, $\ell'$ is divisible by $p((2^{s}(r-1+x)+c-2+y)z+i_{x,y})$ for every $x=0,1$ and $y=0,1,2$.
\end{enumerate}
The table encoding an accepting computation of $N$ with a six-tuple $(r,c)$ is illustrated in Fig.~\ref{fig:tape}.

\begin{figure}[htb]\centering
\includegraphics[scale=1]{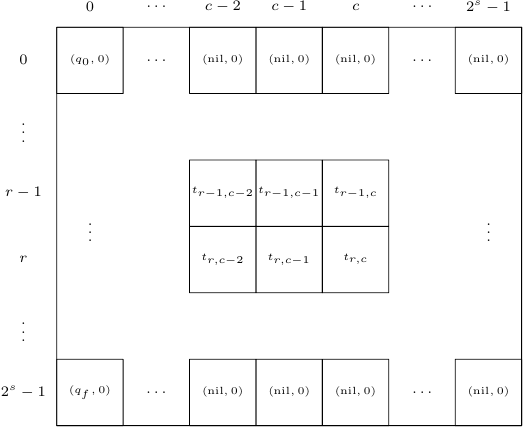}
\caption{The table of tape configurations encoded by $\ell'$ with a six-tuple at $(r,c)$.}\label{fig:tape}
\end{figure}

The final formula $\varphi(\ell')$ is the disjunction of sub-formulas $\varphi_{1}(\ell')$, $\varphi_{2}(\ell')$, $\varphi_{3}(\ell')$, $\varphi_{4}(\ell')$, $\varphi_{5}(\ell')$, $\varphi^{(i_{1,1},i_{1,2},\ldots,i_{2,3})}_{6'}(\ell')$, respectively for each of the six conditions.
Then we simply return
\[
\varphi(\ell') := \bigvee_{1\leq i<j\leq z}\varphi^{i,j}_{1}(\ell') \vee \varphi_{2}(\ell') \vee \varphi_{3}(\ell') \vee \varphi_{4}(\ell') \vee \varphi_{5}(\ell') \bigvee_{(i_{1,1},i_{1,2},\ldots,i_{2,3})\notin W} \varphi^{(i_{1,1},i_{1,2},\ldots,i_{2,3})}_{6'}(\ell')
.\]
Each of the constructed formulas is in the required form for the verifying procedure from Subsection~\ref{subsec:divisibility}, i.e., it is $\exists_{X_1,\ldots,X_k \in \{0,\ldots,2^m-1\}}\; \psi(X_1,\ldots,X_k)$, where $\psi$ uses conjunctions and disjunctions and simple propositions (1)--(6), and where the value of $m$ depends only on $s$ (but the number $k$ of existentially quantified variables might be different in different formulas).
A disjunction of all constructed formulas can be easily rewritten to also have such form.

We only describe in detail how to construct the formula $\varphi^{i,j}_{1}(\ell')$ that is satisfied when, for some $(r,c)$, $\ell'$ is divisible by both $p((2^{s}r+c)z+i)$ and $p((2^{s}r+c)z+j)$.
Formulas $\varphi_{2}(\ell')$, $\varphi_{3}(\ell')$, $\varphi_{4}(\ell')$, $\varphi_{5}(\ell')$, $\varphi^{(i_{1,1},i_{1,2},\ldots,i_{2,3})}_{6'}(\ell')$ are constructed using a similar reasoning.

To construct $\varphi^{i,j}_{1}(\ell')$ we need to bound the primes used to represent the choice.
\begin{lemma}\label{lem:primes}
$p(2^{2s}z) \leq 2^{11s}$.
\end{lemma}
\begin{proof}
A well-known bound \cite{Rosser&Schoenfeld:1962} states that
$p(n) < n(\log n + \log \log n)$ for $n \geq 6$.
Therefore, for all $n\geq 1$ we have $p(n) \leq 2n^{2}$.
Then, $p(2^{2s}z) = p(2^{2s}(2s+2)) \leq p(2^{2s}\cdot 2^{2s+1}) = p(2^{4s+1}) \leq 2^{8s+3} \leq 2^{11s}$.
\end{proof}

Therefore, we set $m=11s$ and quantify variables over $\{0,\ldots,2^{11s}-1\}$.
The expression $\varphi^{i,j}_{1}(\ell')$ is of the form
$\exists_{X_1,\ldots,X_k \in \{0,\ldots,2^{11s}-1\}}\; \psi^{i,j}_{1}(X_1,\ldots,X_k)$,
where $\psi^{i,j}_{1}(X_1,\ldots,X_k)$ is a conjunction of simple propositions.
For clarity, we construct it incrementally.

Recall that we are looking for $r,c \in \{0,\ldots,2^{s}-1\}$.
We start with quantifying over $r,c,r',c',\Delta \in \{0,\ldots,2^{11s}-1\}$ and including $(\Delta = 2^{10s}) \wedge (r\cdot \Delta = r') \wedge (c\cdot \Delta = c')$ in the conjunction. This guarantees that indeed $r,c\in \{0,\ldots,2^{s}-1\}$.

Then we quantify over $z_{r},z_{c},m_{r},m_{c},m_{s}\in \{0,\ldots,2^{11s}-1\}$ and include the following in the conjunction:
\[
(z_{r} = 2^{s}\cdot z) \wedge (z_{c} = z) \wedge (z_{r} \cdot r = m_{r}) \wedge (z_{c} \cdot c = m_{c}) \wedge (m_{s} = m_{r} + m_{c})
.\]
This ensures that $m_{s} = (2^{s}\cdot r + c)z$.
Additionally, we quantify over $m_{i},m_{j},m'_{i},m'_{j}\in \{0,\ldots,2^{11s}-1\}$ and add:
\[
(m_{i} = i) \wedge (m_{j} = j) \wedge (m_{s} + m_{i} = m'_{i}) \wedge (m_{s} + m_{j} = m'_{j})
\]
to the conjunction, which ensures that $m'_i = (2^{s}\cdot r + c)z+i$ and $m'_j = (2^{s}\cdot r + c)z+j$.
Finally, we quantify over $p_{i},p_{j}\in \{0,\ldots,2^{11s}-1\}$ and include:
\[
(p_{i} \text{ is the $m'_{i}$'th prime number}) \wedge (p_{j} \text{ is the $m'_{j}$'th prime number}) \wedge (p_{i} \mid \ell') \wedge (p_{j} \mid \ell')
\]
in the conjunction.
By construction, the obtained $\varphi^{i,j}_{1}(\ell')$ is satisfied when, for some $(r,c)$, $\ell'$ is divisible by both $p((2^{s}r+c)z+i)$ and $p((2^{s}r+c)z+j)$.

The other formulas are constructed using the same principle, and then we rewrite their disjunction to have the required form and obtain $\varphi(\ell')$.
Formula $\varphi(\ell')$ is satisfied exactly when $\ell'$ does not correspond to an accepting computation of $N$.
Therefore, we have reduced (by a polynomial reduction) the canonical NEXPTIME-complete to Problem~\ref{pbm:ns-edf}, which was reduced to Problem~\ref{pbm:elu}.

Since a standard binarization works for NFAs, we can further reduce to the binary case.
We state our final result:
\begin{theorem}
Existential length universality (Problem~\ref{pbm:elu}) for NFAs is NEXPTIME-hard, even if the alphabet is binary.
\end{theorem}

\section{The case where $M$ is a regular expression}\label{sec:re}

\begin{theorem}\label{thm:re_given_length}
Specified-length universality (Problem~\ref{pbm:glu}) for regular expressions and NFAs is PSPACE-complete, and existential length universality (Problem~\ref{pbm:elu}) for regular expressions is PSPACE-hard.
\end{theorem}
\begin{proof}
We transform a regular expression $M$ into an NFA $N$ with linear number of states in the length of $M$.
To see that the problem is in nondeterministic linear space, note that to verify that
$\Sigma^\ell \subsetneq L(M)$, all we need do is guess a string of length $\ell$ that is \emph{not} accepted and then simulate $N$ on this string.
Of course, we do not store the actual string, we just guess it symbol-by-symbol, meanwhile counting up to $\ell$ to make sure the length is correct.
The counter can be achieved in $\O(\log \ell)$ space. By Savitch's theorem, it follows that the problem is in PSPACE.

To see that the problem is PSPACE-hard, we model our proof after the proof that non-universality for NFA's in PSPACE-hard (see, e.g., \cite[Section~10.6]{Aho&Hopcroft&Ullman:1974}).
That proof takes a polynomial-space-bounded deterministic TM $T$ and input $x$, and creates a polynomial length regular expression $M$ that describes all strings which do not correspond to accepting computations.
So $L(M) \neq \Sigma^*$ if and only if $T$ accepts $x$.
Strings can fail to correspond to accepting computations because they begin wrong, end wrong, have a syntax error, have an intermediate step that exceeds the polynomial-space bound, or have a transition that does not correspond to a rule of the TM.
Our regular expression is a sum of these cases.

We now modify this construction.
First, we assume our TM always has the next move, except out of the halting state $h$, where there is no valid next move.
This is easy to ensure by introducing a few special states that yield a cycle and do not bring the head beyond the polynomial space bound; the TM goes into these states when the original TM does not have a transition.
Thus on every input, we either halt or loop forever (within the polynomial space bound).
Next, we create a regular expression $M$ describing the language of invalid computations; its language contains all strings that begin wrong, have syntax errors, or fail to follow the rules of the TM, but no strings that are a prefix of a possibly-accepting computation.
This shows that existential length universality is PSPACE-hard.

For the hardness of specified-length universality, we set $\ell$ to be the length of words describing computations of a length longer than the one obtained from the polynomial space bound for $T$, that is, $2^{|x|^c}+1$.
So if $T$ fails to accept an input $x$, for every length there is some prefix of a computation that is correct, which is not in $L(M)$.
On the other hand, if $T$ accepts $x$, then every string of length at least $\ell$ is either syntactically incorrect or has a bogus transition -- including a transition out of the halting state -- so every string of that length is in the language.
\end{proof}

The problem about the exact complexity of existential length universality in the case of a regular expression remains open, as well as the bounds for the minimal universality length.

\begin{openquestion}
What is the complexity of existential length universality (Problem~\ref{pbm:elu}) when $M$ is a regular expression?
What is the largest possible minimal universality length in terms of the length of $M$?
\end{openquestion}

If we could show that the minimal universality length is at most exponential, then it would immediately imply that the problem is PSPACE-complete.
The best example that we have found has exponential minimal universality length.

\begin{theorem}
There exists a binary regular expression with $n$ input symbols for which the minimal universality length is $2^{\Omega(n)}$.
\end{theorem}
\begin{proof}
We already know that there is a class of regular expressions such that the shortest not specified string looks like the binary expansions of $0, 1, ..., 2^m - 1$ separated by delimiters.
This can be found, for example, in~\cite[Section~5]{Ellul:2005}.
More precisely, the shortest not specified string looks like (for $m=3$):
\begin{verbatim}
#000#100#010#110#001#101#011#
 100 010 110 001 101 011 111
\end{verbatim}
\noindent where columns denote composite symbols and the bottom row is one more than the
top row and numbers are written least-significant-bit-first.

Now we can modify this construction so not only does it not accept this
string, but it also fails to accept every prefix of this string.
This is easy, as it just involves deleting the conditions about ending
properly (so condition~(2) from~\cite[Section~5]{Ellul:2005}).

So our regular expression is a sum of two parts.
One specifies all strings except those that are prefixes of the string above;
the other one specifies all strings that end in the proper string, which in for $m=3$ would be
\begin{verbatim}
#011#
 111
\end{verbatim}

Since no prefix ends with this except the total string itself, this regular expression will fail to specify all strings of every length until it gets to the length of the total string ($(m+1)(2^m-1)$).
As in \cite{Ellul:2005}, the regular expression can be converted to a binary one with $n$ input symbols, where $n$ is linear in $m$.
Thus, its minimal universality length is $2^{\Omega(n)}$.
\end{proof}

\section{The case where $M$ is a PDA}\label{sec:pda}

\begin{theorem}\label{thm:PDA_unsolvable}
Existential length universality (Problem~\ref{pbm:elu}) is recursively unsolvable for PDAs.
\end{theorem}
\begin{proof}
We modify the usual proof \cite[Thm.~8.11, p.~203]{Hopcroft&Ullman:1979}
that ``Given a PDA $M$, is $L(M)  = \Sigma^*$?'' is an unsolvable problem.

In that proof, we start with a Turing machine $T$, assumed to have a single halt state $h$ and no transitions out of this halting state.
We consider the language $L$ of all valid accepting computations of $T$.
A valid computation consists of a sequence of configurations of the machine, separated by a delimiter $\#$.
Every second configuration is reversed, i.e., the configurations are written unreversed and reversed alternatingly.
Each unreversed configuration is of the form $x q y$, where $xy$ is the TM's tape contents, $q$ is the TM's current state, and the TM is scanning the first symbol of $y$.  
A valid computation must start with $\# q_0 x \#$ for some string $x$, and end with $\# y h z \#$ for some strings $y$ and $z$, where $h$ is the halt state.
Furthermore, two consecutive configurations inside a valid computation must follow by the rules of the TM.  

Thus we can accept $\overline{L}$ with a PDA by checking (nondeterministically) if a given string begins wrong, ends wrong, is syntactically invalid,
or has two consecutive configurations that do not follow by the rules of $T$.
Only this last requirement presents any challenge.
The idea is to push the first configuration on, then pop it off and compare to the next (this is why every other configuration
needs to be reversed).
Comparing two configurations just requires matching symbols, except in the region of the state, where we must verify that the second configuration follows by the rules of $T$ from the first.

Thus, a PDA can be constructed to accept $\overline{L}$, the set of all strings that {\it do not} represent valid accepting computations of $T$.
Hence ``Given $T$, is $L(T) = \Sigma^*$?'' is unsolvable, because if we could answer it, we would know whether $T$ accepts some string.
This concludes our sketch of the usual universality proof.

Now, to prove our result about existential length universality, we modify the above construction in two ways.
First, we assume that our TM $T$ has the property that there is always a next move possible, except from the halting state.
Thus, the only possibilities on any input are (1) arriving in the halting state $h$, after which there is no move or (2) running forever without halting.

Second, we make a PDA based on a TM $T$ such that our PDA fails to accept all strings that represent valid halting computations of $T$ that start with empty input, and \emph{also} fails to accept all strings that are prefixes of a valid computation.
In other words, our PDA is designed to accept if a computation starts wrong, is syntactically invalid, or if a configuration does not follow from the previous one.
If the last configuration is incomplete, and what is actually present does not violate any rules of $T$, however, our PDA \emph{does not} accept.  

Now, suppose that $T$ does not accept the empty string.
Then $T$ does not halt on empty input, so there are valid computations for every number of steps.
So there are strings representing valid computations, or prefixes of valid computations, of every length.
Since $M$ fails to accept these, there is no $\ell$ such that $M$ accepts all strings of length $\ell$.

On the other hand, if $T$ does accept the empty string, then there is exactly one valid halting computation for it; say it is of length $s$, where the last configuration is of length $t$.
Consider every putative computation of length $\geq s+t+2$; it must violate
a rule, since there are no computations out of the halting state.
So $M$ accepts all strings of length $s+t+2$.
Thus we could decide if $T$ accepts the empty string, which means that existential length universality for PDAs is unsolvable.
\end{proof}

\begin{corollary}
Fix an alphabet $\Sigma$.
For a PDA $M$, let $\ell(M)$ be the minimal universality length (if it exists) of $L(M)$.
Let $f(n)$ be the maximum of $\ell(M)$ over all PDA's $M$ of size $n$.
Then $f(n)$ grows faster than every computable function.
\end{corollary}
\begin{proof}
If $f(n)$ grew slower than some computable function $g(n)$, then we could
solve existential universality for a given PDA $M$ of size $n$ by
(1) computing $g(n)$;
(2) testing, by a brute-force enumeration, whether $M$ accepts all strings of length $\ell$ for all $\ell \le g(n)$.
(We can deterministically test if a PDA accepts a string by converting the
PDA to an equivalent context-free grammar and then using the usual
Cocke-Younger-Kasami dynamic programming algorithm for CFL membership.)
But existential length universality is undecidable for PDA's.
\end{proof}

In contrast, the specified-length universality problem is solvable in the case of a PDA.
\begin{proposition}
Specified-length universality for PDAs is in co-NEXPTIME.
\end{proposition}
\begin{proof}
We can convert a PDA in polynomial time to a context-free grammar generating the same language.
Then we can guess an (exponentially long in the size of the input) word of length $\ell$ such that it is not accepted, and verify if it is indeed not generated by the context-free grammar (using, for example, the Cocke-Younger-Kasami dynamic programming solution).
\end{proof}

Finally, we prove the co-NEXPTIME-hardness.

A tiling system is a $\mathcal{T} = (T,H,V,t_I,t_F)$ where $T$ is a finite set.
Its elements are called \emph{tile types} or just \emph{tiles}.
$H,V \subseteq T^2$ are called the \emph{horizontal}, resp., \emph{vertical} matching relation; $t_I,t_F \in T$ are two designated \emph{initial} and
\emph{final} tiles.
Given an $n \in \mathbb{N}$, a \emph{$\mathcal{T}$-tiling of the $n$-th exponential square} is a function $\tau\colon \{0,\ldots,2^n-1\} \times \{0,\ldots,2^n-1\} \to T$.
It is called \emph{valid} if 
\begin{itemize}
\item $\tau(0,0) = t_I$ and $\tau(2^n-1,2^n-1) = t_F$,
\item for all $i=0,\ldots,2^n-2$ and all $j=0,\ldots,2^n-1$: $(\tau(i,j),\tau(i+1,j)) \in H$, and
\item for all $i=0,\ldots,2^n-1$ and all $j=0,\ldots,2^n-2$: $(\tau(i,j),\tau(i,j+1)) \in V$.
\end{itemize}

\begin{proposition}[\cite{Boas:1997}]
The following problem is NEXPTIME-hard:
Given a tiling system $\mathcal{T}$ and a number $n$ encoded unarily, is there a valid tiling of the $n$-th exponential square?
\end{proposition}

For a given $T$ we define two sets of symbols via $\overrightarrow{T} := \{ \overrightarrow{t} \mid t \in T \}$ and
$\overleftarrow{T} := \{ \overleftarrow{t} \mid t \in T \}$. Then let 
$\Sigma_T := \overrightarrow{T} \cup \overleftarrow{T} \cup \{0,1,\overrightarrow{\#},\overleftarrow{\#}\}$.
We use $\bar{T}$ to abbreviate $\overrightarrow{T} \cup \overleftarrow{T}$, $\bar{\#}$ to abbreviate
$\{\overrightarrow{\#},\overleftarrow{\#}\}$, and $\mathsf{B}$ to abbreviate $\{0,1\}$.

Given a $\mathcal{T}$-tiling $\tau$ of the exponential square for parameter $n$, we define the
\emph{encoding} of the $j$-th row as the following word $w_j^\tau$ over the alphabet $\Sigma_T$.
\begin{displaymath}
w_j^\tau := s_j\, \overrightarrow{j}\, \overrightarrow{\tau(0,j)}\, \overrightarrow{0} \, \overleftarrow{\tau(1,j)} \, \overleftarrow{1} \, \overrightarrow{\tau(2,j)} \overrightarrow{2} \ldots 
  \overleftarrow{\tau(2^n{-}1,j)}\,\overleftarrow{2^n{-}1} 
\end{displaymath}
where $s_j = \overrightarrow{\#}$ if $j$ is even and $s_j = \overleftarrow{\#}$ if $j$ is odd. Moreover, $\overrightarrow{i}$ denotes
the $n$-bit binary coding of the number $i$ with least significant bit on the 
left, and $\overleftarrow{i}$ does the same but with least significant bit on the right. Likewise, for a string 
$u \in \{0,1\}^+$ we write $\mathit{val}^{\rightarrow}(u)$, resp., $\mathit{val}^{\leftarrow}(u)$ to denote the number
that is encoded by $u$ with least significant bit on the left, resp., on the right.

The \emph{encoding} of the entire tiling $\tau$ is the word
\begin{displaymath}
w^\tau := w_0^\tau\, w_1^\tau\, w_2^\tau \, \ldots \, w_{2^n-1}^\tau
\end{displaymath}
Intuitively, the encoding of such a tiling lists the tiles placed in the square row-wise.
Each row is preceded by the row number, and each tile is followed by its column number in binary coding.
The order for the bits of even row and column numbers is that of increasing significance; for odd ones they decrease in significance from left to right.

\begin{lemma} 
\label{lem:length}
Let $\tau$ be a tiling of the $n$-th exponential square. Then $|w^\tau| = (n+1)\cdot (2^{2n} + 2^n)$.
\end{lemma}

\begin{lemma}
\label{lem:construct}
Given a tiling system $\mathcal{T}$ and a number $n$ encoded unarily. There is a PDA that recognizes all words 
which are not encodings of a valid $\mathcal{T}$-tiling of the $n$-th square.
\end{lemma}

\begin{proof}
A word $w$ over the alphabet $\Sigma_T$ is not the encoding of a valid $\mathcal{T}$-tiling only if 
one of the following conditions holds.
\begin{enumerate}
\item $w$ is not of the form 
\begin{displaymath}
\overrightarrow{\#}0^n (\overrightarrow{T}0^n(\bar{T}\mathsf{B}^n)^*)\overleftarrow{T}1^n)
(\bar{\#}\mathsf{B}^n (\overrightarrow{T}0^n(\bar{T}\mathsf{B}^n)^*)\overleftarrow{T}1^n))^*
\overleftarrow{\#}1^n (\overrightarrow{T}0^n(\bar{T}\mathsf{B}^n)^*)\overleftarrow{T}1^n)
\end{displaymath}
\item There is a subword $\overrightarrow{\#}u \ldots \overleftarrow{\#}v$ with $|u|=|v|=n$ and the part in between does not 
      contain any $\bar{\#}$-symbol but $\mathit{val}^{\leftarrow}(v) \ne \mathit{val}^\rightarrow(u)+1$. Likewise for
      subwords of the form $\overleftarrow{\#}u \ldots \overrightarrow{\#}v$.
\item There is a subword of the form $u \overleftarrow{T}v$ with $|u|=|v|=n$ but 
      $\mathit{val}^{\leftarrow}(v) \ne \mathit{val}^\rightarrow(u)+1$. Likewise for subwords of the form $u\overrightarrow{T}v$.
\item The symbol at position $n+1$ (starting with 0) of $w$ is not $\overrightarrow{t_I}$.
\item There is no occurrence of a subword of the form $\overleftarrow{\#}1^n \ldots \overleftarrow{t_F} 1^n$ with no further
      $\bar{\#}$-symbols in between.
\item There is a subword of the form $\overrightarrow{t}u\overleftarrow{t'}$ with $u \in \mathsf{B}^n$ and $(t,t') \not\in H$.
      Likewise for subwords of the form $\overleftarrow{t}u\overrightarrow{t'}$.
\item There is a subword of the form $\overrightarrow{t}u\ldots \overrightarrow{t'}u$ with $u \in \mathsf{B}^n$, a single occurrence
      of a $\bar{\#}$-symbol in between, and $(t,t') \not\in V$. Likewise for $\overleftarrow{t}u\ldots \overleftarrow{t'}u$.
\end{enumerate}
Conditions 1--3 check whether the word is an encoding of a tiling according to the definition above. Conditions 4--7 check whether
it is valid.

Conditions 1, 4, 5 and 6 are easily seen to be regular. It is not difficult to construct a PDA making use of nondeterminism for each
of the conditions 2, 3 and 7. For example, condition 2 can be checked by guessing a position containing $\overrightarrow{\#}$, then
pushing the following $u \in \mathsf{B}^n$ onto the stack, then moving to the next occurrence of $\overleftarrow{\#}$ and then
comparing the following $v$ with the stack content in the usual way for binary arithmetic: a prefix of $v$ must equal the top of
the stack until a $1$ in the input is read and the corresponding symbol on the stack is a $0$. The rest of $v$ must contain $0$'s
only while the rest of the stack contains $1$'s only.

Since PDAs are closed under union and NFAs are PDAs, the automata hinted at can be used to form a single PDA 
$\mathcal{A}_{\mathcal{T},n}$ that accepts a word if and only if this word is not a valid $\mathcal{T}$-tiling of the $n$-th exponential square.
\end{proof}

\begin{theorem}
Specified-length universality for PDA is co-NEXPTIME-hard, even if the alphabet is binary.
\end{theorem}
\begin{proof}
By reduction from the exponential-square tiling problem. Given $\mathcal{T},n$ with $n$ encoded unarily, it is not hard to see that 
the PDA $\mathcal{A}_{\mathcal{T},n}$ of the previous lemma can be constructed in polynomial time in $|\mathcal{T}|+n$. Furthermore,
let $n' := (n+1)\cdot (2^{2n} + 2^n)$. Note that the binary encoding of $n'$ can be constructed from $n$ in time polynomial in the
value of $n$ which is sufficient here since $n$ is encoded unarily.

It remains to see that this reduction is correct, i.e., that there is a valid $\mathcal{T}$-tiling of the $n$-th exponential square
if and only if $L(\mathcal{A}_{\mathcal{T},n})$ does not contain all $\Sigma_T$-words of length $n'$.  

$\Rightarrow$: Suppose there is a valid $\mathcal{T}$-tiling $\tau$. According to Lemma~\ref{lem:length}, $|w^\tau| = n'$. 
According to Lemma~\ref{lem:construct} we have $w^\tau \not\in L(\mathcal{A}_{\mathcal{T},n})$.

$\Leftarrow$: Suppose that there is some $w \in (\Sigma_T)^{n'} \setminus L(\mathcal{A}_{\mathcal{T},n})$. By the construction
of $\mathcal{A}_{\mathcal{T},n}$, $w$ cannot meet any of the conditions set out in the proof of Lemma~\ref{lem:construct}. In other
words, it meets all of their positive counterparts. Because of the conditions (1)--(3), $w$ must, in fact, encode a $\mathcal{T}$-tiling 
of the $n$-th exponential square. Because of conditions (4)--(7), this tiling must be valid. 

Finally, a standard binarization applies to PDAs.
The length $|w^\tau|$ stated in Lemma~\ref{lem:length} can easily be adjusted accordingly.
\end{proof}

\bibliographystyle{plain}
\bibliography{lu}
\end{document}